\documentclass[sigconf, nonacm]{aamas}

\usepackage{balance}

\usepackage{graphicx}
\usepackage{latexsym}

\usepackage[utf8]{inputenc} 
\usepackage[T1]{fontenc}    
\usepackage{hyperref}       
\usepackage{url}            
\usepackage{booktabs}       
\usepackage{nicefrac}       
\usepackage{microtype}      
\usepackage{xcolor}         
\usepackage{wrapfig}
\usepackage{istgame}
\usepackage{booktabs}
\usepackage{xpatch}
\usepackage{amsmath}
\usepackage{amsthm}
\usepackage{mathtools}
\usepackage[shortcuts]{extdash}
\usepackage{xcolor}
\usepackage{url}
\usepackage{algorithm}
\usepackage{algpseudocode}
\usepackage{multirow}

\setcopyright{none}
\acmConference[ALA '24]{Proc.\@ of the Adaptive and Learning Agents Workshop (ALA 2024)}
{May 6-7, 2024}{Online, \url{https://ala2024.github.io/}}{Avalos, Müller, Wang, Yates (eds.)}
\copyrightyear{2024}
\acmYear{2024}
\acmDOI{}
\acmPrice{}
\acmISBN{}
\acmSubmissionID{13}

\title{Continual Depth-limited Responses for Computing Counter-strategies in Sequential Games}

\author{David Milec}
\affiliation{
  \institution{AI Center, FEE, CTU in Prague}
  \country{Czech Republic}}    
\email{milecdav@fel.cvut.cz}

\author{Ondřej Kubíček}
\affiliation{
  \institution{AI Center, FEE, CTU in Prague}
  \country{Czech Republic}}
\email{kubicon3@fel.cvut.cz}

\author{Viliam Lisý}
\affiliation{
  \institution{AI Center, FEE, CTU in Prague}
  \country{Czech Republic}}
\email{viliam.lisy@agents.fel.cvut.cz}

\begin{abstract}
In zero-sum games, the optimal strategy is well-defined by the Nash equilibrium. However, it is overly conservative when playing against suboptimal opponents and it can not exploit their weaknesses. Limited look-ahead game solving in imperfect-information games allows superhuman play in massive real-world games such as Poker, Liar's Dice, and Scotland Yard. However, since they approximate Nash equilibrium, they tend to only win slightly against weak opponents. We propose theoretically sound methods combining limited look-ahead solving with an opponent model, in order to 1) approximate a best response in large games or 2) compute a robust response with control over the robustness of the response. Both methods can compute the response in real time to previously unseen strategies. We present theoretical guarantees of our methods. We show that existing robust response methods do not work combined with limited look-ahead solving of the shelf, and we propose a novel solution for the issue. Our algorithm performs significantly better than multiple baselines in smaller games and outperforms state-of-the-art methods against SlumBot.
\end{abstract}


\DeclareMathOperator*{\argmax}{arg\,max}

\begin{document}

\pagestyle{fancy}
\fancyhead{}

\newtheorem{observation}{Observation}

\newcommand{\Iset}{\mathcal{I}}
\newcommand{\pr}{1}
\newcommand{\ps}{2}
\newcommand{\sbr}{\mathcal{B}}
\newcommand{\srnr}{\mathcal{R}}
\newcommand{\partition}{\mathcal{P}}
\def\real{\mathbb{R}}

\newcommand\todo[1]{\textcolor{red}{TODO: #1}}

\def\printappendix{1}
\def\printmain{1}
\def\printacknowledgments{1}
\if\printmain1

    \maketitle

    \section{Introduction}
We can not enumerate all the decision points in large games, which makes computing optimal strategy, a Nash equilibrium (NE) in two-player zero-sum games, infeasible. A breakthrough that allowed approximating the NE and defeating human experts in several large imperfect-information games is limited look-ahead solving or search, which adapts the well-known approach from perfect-information games to games with imperfect-information \cite{moravvcik2017deepstack,brown2019superhuman,schmid2021player}. Limited look-ahead solving takes advantage of decomposition. It iteratively builds the game to some depth and solves a small part of the game while summarising the required values from the rest of the game by a value function. The value function is commonly learned using neural networks. When the algorithms solve the game step by step, it is called continual depth-limited solving or continual resolving.

The vast majority of theoretically sound, continual depth-limited solving algorithms assume perfect rationality of the opponent and do not allow explicit modeling of an opponent and exploitation of the opponent's mistakes. As a result, even very weak opponents exploitable by the heuristic local best response (LBR) \cite{lisy2017eqilibrium} can tie or lose very slowly against these methods \cite{zarick2020unlocking}. Therefore, there has been a significant amount of work towards computing strategies to use against imperfect opponents to create AI systems that would perform well in the real world, for example, against humans \cite{bard2013online,wu2021l2e,southey2012bayes,mealing2015opponent,korb2013bayesian,milec2021complexity,johanson2009data}.

The opponent modeling and exploitation process consists of two steps: opponent modeling and model exploitation. Opponent modeling requires building a model from previous data or actions observed during an online play. Model exploitation is finding a good strategy against the given model and is the main focus of this paper. In smaller games, we can trivially compute a best response to exploit the opponent maximally, or we can use methods to compute robust responses \cite{johanson2008computing, johanson2009data} if there is uncertainty in the model and we want to be safer, meaning we want to limit the possible loss when facing the worst-case adversary. However, even the best response (BR) computation in large games is non-trivial, and currently, no approach can compute it while interacting in real-time.

This work explores the full model exploitation and proposes continual depth-limited best response (CDBR). CDBR relies on the value function used in the standard limited look-ahead solving, and we prove theoretical guarantees on the performance. A drawback of using the same value function is decreased performance, and we could improve CDBR by training a specific value function for a particular opponent model. However, it would be impractical since the training is expensive. Furthermore, in cases where we learn the opponent model in real-time interaction and update it after each step, it would be impossible.

The best response and CDBR are useful, e.g., for evaluating the quality of strategies, but they are brittle in game play. We can lose significantly when facing an opponent different from the expected model. In the real world, we will never have exact models, which makes BR and CDBR impractical for game play. To address the issue, robust responses are used \cite{ganzfried2015safe, johanson2008computing, johanson2009data}. They introduce a notion of safety, and the safety criterion requires the response to stay close to the NE. In other words, only to lose a limited amount to the worst-case adversary. Trivially, we can compute both BR and NE and create a linear combination where we can control the safety by a parameter. However, previous work shows that we can perform significantly better and recover the whole Pareto set of maximally exploiting strategies with maximal safety \cite{johanson2008computing}. We adapt the method to limited look-ahead solving, creating a continual depth-limited restricted Nash response (CDRNR). Similarly to the full robust response, CDRNR significantly outperforms the linear combination. However, it comes with drawbacks in the limited look-ahead solving. Namely, we need to keep the previously solved subgames as a path to the root to ensure theoretical soundness, which linearly increases the size of the game solved each step, making it scalable to games with low depth like Poker or Goofspiel but impractical in games with high depth.

Our contributions are: \textbf{1)}  We formulate the algorithms to find the responses given the opponent strategy and an evaluation function. This results in the best performing theoretically sound robust response applicable to large games. \textbf{2)} We prove the soundness of the proposed algorithms. \textbf{3)} We provide an analysis of problems that arise when using opponent models in limited look-ahead solving and propose a solution we call a full gadget. \textbf{4)} We empirically evaluate the algorithms on poker and goofspiel variants and compare them to multiple baselines. We show that our responses exploit the opponents, and CDBR outperforms domain-specific local best response \cite{lisy2017eqilibrium} on poker. We also compare CDBR with the approximate best response (ABR) on smaller games and on full Heads-up No-Limit Texas Hold'em (HUNL), where we exploit SlumBot significantly more than ABR.
    \section{Background}
\label{sec:background}
A two-player extensive-form game (EFG) consists of a set of players $N = \{\pr,\ps,c\}$, where $c$ denotes the chance, $\pr$ is the maximizer and $\ps$ is the minimizer, a finite set $A$ of all actions available in the game, a set $H \subset \{a_1 a_2 \cdots a_n \mid a_j \in A, n \in \mathbb{N}\}$ of histories in the game. We assume that $H$ forms a non-empty finite prefix tree. We use $g \sqsubset h$ to denote that $h$ extends $g$. The \textit{root} of $H$ is the empty sequence $\emptyset$. The set of leaves of $H$ is denoted $Z$, and its elements $z$ are called \textit{terminal histories}. The histories not in Z are \textit{non-terminal histories}. By $A(h) = \{a \in A \mid ha \in H\}$, we denote the set of actions available at $h$. $P : H \setminus Z \to N$ is the \textit{player function} which returns who acts in a given history. Denoting $H_i = \{h \in H \setminus Z \mid P(h) = i\}$, we partition the histories as $H = H_\pr \cup H_\ps \cup H_c \cup Z$. $\sigma_c$ is the \textit{chance strategy} defined on $H_c$. For each $h \in H_c, \sigma_c(h)$ is a fixed probability distribution over $A(h)$. Utility functions assign each player utility for each leaf node, $u_i : Z \to \mathbb{R}$. The game is zero-sum if $\forall z \in Z: \quad u_\pr(z) + u_\ps(z) = 0$. In the paper, we assume all the games are zero-sum. The game is of \textit{imperfect information} if all players do not fully observe some actions or chance events. The information structure is described by \textit{information sets} for each player $i$, which forms a partition $\Iset_i$ of $H_i$. For any information set $I_i \in \Iset_i$, any two histories $h, h' \in I_i$ are indistinguishable to player $i$. Therefore $A(h) = A(h')$ whenever $h, h' \in I_i$. For $I_i \in \Iset_i$ we denote by $A(I_i)$ the set $A(h)$ and by $P(I_i)$ the player $P(h)$ for any $h \in I_i$.

A \textit{strategy} $\sigma_i \in \Sigma_i$ of player $i$ is a function that assigns a distribution over $A(I_i)$ to each $I_i \in \Iset_i$. A \textit{strategy profile} $\sigma = (\sigma_\pr, \sigma_\ps)$ consists of strategies for both players. $\pi^\sigma(h)$ is the probability of reaching $h$ if all players play according to $\sigma$. We can decompose $\pi^\sigma(h) = \prod_{i \in N}\pi^\sigma_i(h)$ into each player's contribution. Let $\pi^\sigma_{-i}$ be the product of all players' contributions except that of player $i$ (including chance). For $I_i \in \Iset_i$ define $\pi^\sigma(I_i) = \sum_{h \in I_i}\pi^\sigma(h)$, as the probability of reaching information set $I_i$ given all players play according to $\sigma$. $\pi_i^\sigma(I_i)$ and $\pi_{-i}^\sigma(I_i)$ are defined similarly. Finally, let $\pi^\sigma(h,z) = \frac{\pi^\sigma(z)}{\pi^\sigma(h)}$ if $h \sqsubset z$, and zero otherwise. $\pi^\sigma_i(h,z)$ and $\pi^\sigma_{-i}(h,z)$ are defined similarly. Using this notation, \textit{expected payoff} for player $i$ is $u_i(\sigma) = \sum_{z \in Z}u_i(z)\pi^\sigma(z)$. A \textit{best response} (BR) of player $i$ to the opponent's strategy $\sigma_{-i}$ is a strategy $\sigma_i^{BR} \in BR_i(\sigma_{-i})$, where $u_i(\sigma_i^{BR}, \sigma_{-i}) \geq u_i(\sigma'_i, \sigma_{-i})$ for all $\sigma'_i \in \Sigma_i$. A tuple of strategies $( \sigma_i^{NE}, \sigma_{-i}^{NE})$, $\sigma_{i}^{NE}\in\Sigma_{i}, \sigma_{-i}^{NE}\in\Sigma_{-i}$ is a \emph{Nash Equilibrium} (NE) if $\sigma_i^{NE}$ is an optimal strategy of player $i$ against strategy $\sigma_{-i}^{NE}$. Formally: $\sigma_i^{NE} \in BR(\sigma_{-i}^{NE})\quad\forall i\in \{\pr,\ps \}$.

In a two-player zero-sum game, the \textbf{exploitability} of a strategy is the expected utility a fully rational opponent can achieve above the value of the game. Formally, exploitability $\mathcal{E}(\sigma_i)$ of strategy $\sigma_i\in\Sigma_i$ is
$
    \mathcal{E}(\sigma_i) =  u_{-i}(\sigma_i, \sigma_{-i}) - u_{-i}(\sigma^{NE}), \quad \sigma_{-i} \in BR_{-i}(\sigma_i).
$

\textbf{Safety} is defined based on exploitability and $\epsilon$-safe strategy is a strategy which has exploitability at most $\epsilon$.

We define \textbf{gain} of a strategy against a model as the expected utility we receive above the value of the game. We formally define the gain $\mathcal{G}(\sigma_i, \sigma_{-i})$ of the strategy $\sigma_i$ against a strategy $\sigma_{-i}$ as
$
\mathcal{G}(\sigma_i, \sigma_{-i}) = u_i(\sigma_i, \sigma_{-i}) - u_{i}(\sigma^{NE}).
$

\begin{figure}
    \centering
    \includegraphics[width=\linewidth]{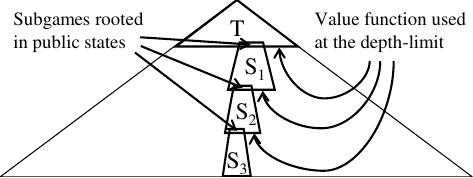}
    \caption{Illustration of the depth-limited solving.}
    \label{fig:subgames}
\end{figure}

\paragraph{Depth-limited Solving - Figure~\ref{fig:subgames}}
We denote $H_i(h)$ the sequence of player $i$'s information sets and actions on the path to a history $h$. Two histories $h, h'$ where player $i$ does not act are in the same \textit{augmented information set} $I_i$ if $H_i(h) = H_i(h')$. We partition the game histories into \textbf{public states} $PS \subset H$, which are closed under the membership within the augmented information sets of all players. \textbf{Trunk} is a set of histories $T \subset H$, closed under prefixes and public states. \textbf{Subgame} $S \subset H$ is a forest of trees with all the roots starting in one public state. It is closed under public states, and the trees can end in terminal public states or often end after a number of moves or rounds in the game. \textbf{Range} of a player $i$ is a probability distribution over his information sets in some public state $PS_i$, given we reached the $PS_i$. \textbf{Value function} is a function that takes the public state and both players' ranges as input and outputs values for each information set in the public state for both players. We assume using an approximation of an \textbf{optimal value function}, which is a value function returning the values of using some NE after the depth-limit. \textbf{Subgame partitioning} $\mathcal{P}$ is a partitioning that splits the game into trunk and subgames into multiple different levels based on some depth-limit or other factors (domain knowledge). Subgame partitioning can be naturally created using the formalism of factored-observation stochastic games \cite{kovavrik2019value}. By $u_i(\sigma)^T_V$, we denote the utility for player $i$ if we use strategy profile $\sigma$ in trunk $T$ and compute values at the depth-limit using value function $V$. When resolving a subgame with just the ranges, there are no guarantees on the resulting exploitability of the strategy in the full game, and the exploitability can rise significantly \cite{burch2014solving}. To address the issue, artificially constructed games called \textbf{gadgets} are used to limit the increase in exploitability. They do it by adding nodes to the top of the subgame, which simulates that the opponent is allowed to deviate from its strategy in an already solved game.

Figure~\ref{fig:cdbr_example_game} shows a simple game illustrating depth-limited solving. The game starts with player~$\ps$ choosing to either play standard biased matching pennies ($p$) or playing his own version of the game ($q$). In the next round, player $\pr$ does not know which game player $\ps$ chose, and he chooses head (H) or tail (T). Then player $\ps$ guesses head ($h_i$) or tail($t_i$), and if he chooses to play the standard version, he receives 2 when correctly guessing head and 1 when correctly guessing tails. Otherwise, the reward is 0. In the modified version, guessing incorrectly gives 2 to the player $\ps$, and guessing correctly gives 1 for heads and -1 for tails.

In the Nash equilibrium of this game, player $\pr$ plays heads with probability $\frac{2}{3}$ and player $\ps$ chooses his own version of the game with probability $\frac{2}{3}$ and follows with only heads. Public states in this game are always the whole levels (rows) since the actions are never observable by both players. When we start depth-limited solving, we create a trunk, which we select as just the root with the choice to play $q$ or $p$. We start solving the trunk using an iterative algorithm, e.g., counterfactual regret minimization (CFR) \cite{zinkevich2008regret}.

We initialize strategy to uniform, which gives us \textit{range} in the next public state ($\frac{1}{2}, \frac{1}{2}$). We give the \textit{range} to the value function, which returns values as if we played equilibrium in the rest of the game. Value function gives us values in the information sets, which translates to the utility of $-\frac{5}{3}$ for playing heads and $-\frac{2}{3}$ for playing tails. We use the values to update regrets in the CFR and perform the next iteration similarly. When we solve the trunk and recover the equilibrium strategy for the first node, we move to a subgame, for example, a game starting in the information set of player $\pr$ and ending after his action. We need to reconstruct what happened earlier. If we replace the already computed strategy with a chance node, which is called unsafe resolving, we are not guaranteed to recover the equilibrium for player $\pr$. Unsafe resolving can produce solutions ranging from heads with probability $\frac{3}{4}$ to $\frac{1}{3}$ but the only equilibrium is heads with probability $\frac{2}{3}$. The situation is fixed using the mentioned gadgets, which allow the opponent to modify their range above the subgame, forcing the other player to play robustly against all the possible ranges and recover the equilibrium.

\def\xscale{1.1}
\begin{figure}
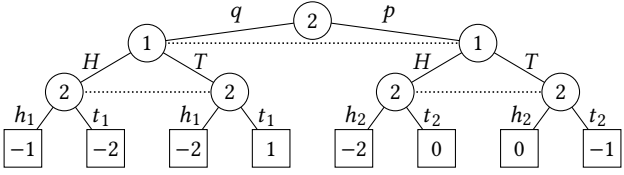

    \begin{istgame}
        \setistOvalNodeStyle{.5cm}
        \setistRectangleNodeStyle{.5cm}
        \xtShowTerminalNodes[rectangle node]
        \istrooto(0){$\ps$}+{3mm}..{4*\xscale cm}+
            \istb{q}[left, xshift=4mm, yshift=2.7mm]  \istb{p}[right, xshift=-4mm, yshift=2.7mm]  \endist
        \istrooto(1)(0-1){$\pr$}+{6.5mm}..{2*\xscale cm}+
            \istb{H}[left, xshift=1.5mm, yshift=1.3mm]  \istb{T}[right, xshift=-1.5mm, yshift=1.3mm]  \endist
        \istrooto(2)(0-2){$\pr$}+{6.5mm}..{2*\xscale cm}+
            \istb{H}[left, xshift=1.5mm, yshift=1.3mm]  \istb{T}[right, xshift=-1.5mm, yshift=1.3mm]  \endist
        \istrooto(3)(1-1){$\ps$}+{7.5mm}..{\xscale cm}+
            \istbt{h\textsubscript{1}}[left, yshift=0.5mm]{-1}[center]  \istbt{t\textsubscript{1}}[right, yshift=0.5]{-2}[center]  \endist
        \istrooto(4)(1-2){$\ps$}+{7.5mm}..{\xscale cm}+
            \istbt{h\textsubscript{1}}[left, yshift=0.5mm]{-2}[center]  \istbt{t\textsubscript{1}}[right, yshift=0.5]{1}[center]  \endist
        \istrooto(5)(2-1){$\ps$}+{7.5mm}..{\xscale cm}+
            \istbt{h\textsubscript{2}}[left, yshift=0.5mm]{-2}[center]  \istbt{t\textsubscript{2}}[right, yshift=0.5]{0}[center]  \endist
        \istrooto(6)(2-2){$\ps$}+{7.5mm}..{\xscale cm}+
            \istbt{h\textsubscript{2}}[left, yshift=0.5mm]{0}[center]  \istbt{t\textsubscript{2}}[right, yshift=0.5]{-1}[center]  \endist
        \xtInfoset(0-1)(0-2)
        \xtInfoset(1-1)(1-2)
        \xtInfoset(2-1)(2-2)
    \end{istgame}    
    \caption{Simple zero-sum imperfect-information game. Nodes denote the decisions of the players, dotted lines mark information sets, and the leaf shows for player \pr}
    \label{fig:cdbr_example_game}
\end{figure}
    \section{Fully Exploiting the Opponent}
Fully exploiting opponent models in small games boils down to computing a best response. This is infeasible in games with an intractable number of information sets for which we use the continual depth-limited solving algorithms. The depth-limited setting does not allow computing BR in one pass anymore. The game we already saw in Figure~\ref{fig:cdbr_example_game} can be an example of that. Suppose we know the player $\ps$ always makes a mistake in the first move and plays only to the standard biased matching pennies. If we knew his strategy of guessing heads or tails, we could compute a best response. However, our trunk will end before the choice, and we need to use a value function. Since the value function in this simple case is just a best response of the opponent, the problem is reduced to finding the optimal strategy against a best response, which corresponds to finding NE, and it can not be solved in one pass. In this section, we propose an algorithm for continual depth-limited best response (CDBR), which generalizes a best response to be used with a value function for depth-limited solving.

\subsection{Continual Depth-limited Best Response}
Given any extensive-form game $G$ with perfect recall, opponent's fixed strategy $\sigma_\ps^F$ and some subgame partitioning $\partition$, we define continual depth-limited best response (CDBR) recursively from the top, see Figure~\ref{fig:subgames}. First, we have trunk $T_1 = T$ and value function $V$. CDBR in the trunk $T_1$ for player $\pr$ with value function $V$ is defined as $\sbr(\sigma_\ps^F)_V^{T_1} = \argmax_{\sigma_\pr}u_\pr(\sigma_\pr, \sigma_\ps^F)_V^{T_1}$. In other words, we maximize the utility over the strategy in the trunk, where we return values from the value function after the depth limit. In each step afterward, we create a new subgame $S_i$ and create new trunk by joining the old one with the subgame, creating $T_i = T_{i-1} \cup S_i$. We fix the strategy of player $\pr$ in the $T_{i-1}$ and maximize over the strategy in the subgame. $\sbr(\sigma_\ps^F)_V^{T_i} = \argmax_{(\sigma_\pr^{S_i})}u_\pr(\sigma_\pr^{S_i} \cup \sigma_\pr^{T_{i-1}}, \sigma_\ps^F)_V^{T_i}$. We continue like that for each step, and we always create a new trunk $T_i$ using the strategy from step $T_{i-1}$ until we reach the end of the game. We denote the full CDBR strategy created by joining strategies from all possible branches $\sbr(\sigma_\ps^F)_V^\mathcal{P}$.

Intuitively, we always solve the game until the depth limit. The opponent is fixed everywhere above the depth limit, and the rational player is fixed in the already solved parts, and she can play in the part that was added last. Looking at Figure~\ref{fig:subgames} CDBR in $S_2$ would allow player $\pr$ to play in $S_2$, it would replace anything bellow $S_2$ with a value function, player $\pr$ would be fixed in $T$ and $S_1$ and player $\ps$ would be fixed in $T, S_1$ and $S_2$.

\subsubsection{Computing CDBR and the complexity}
In practice, we will compute CDBR similarly to depth-limited solving with a few key changes. First, we fix the opponent's strategy in the currently resolved part of the game to allow the player to respond to it, which corresponds to the argmax from the definition. Another key change that simplifies the algorithm is that we no longer need a gadget since the opponent is fixed in the parts we already played through, so we do not need to be robust against different ranges than the one taken from the opponent model.

The difference from the standard depth-limited solving is that we fix the opponent's strategy in the resolved part of the game, and we do not use a gadget. Hence, there is less computation required compared to the standard depth-limited solving.

\subsubsection{Convergence in current iterations}
CFR is an algorithm that needs to track average strategies since the current strategy does not converge to an equilibrium. CFR against best response or a fixed strategy is known to converge in the current strategy \cite{davis2014using,lockhart2019computing}. The next lemma says that CDBR also converges in the current strategy even when a value function is used after the depth-limit.

\begin{lemma}
Let $G$ be a zero-sum imperfect-information extensive-form game. Let $\sigma_\ps^\text{F}$ be the fixed opponent's strategy, and let $T$ be some trunk of the game. If we perform CFR with $t$ iterations in the trunk for player $\pr$, then for the strategy $\hat{\sigma}_{\pr}$ from the iteration with highest expected utility $\max_{\sigma_{\pr}^* \in \Sigma_\pr} u_\pr(\sigma_{\pr}^*, \sigma_\ps^F)_V^T - u_\pr(\hat{\sigma}_{\pr}, \sigma_\ps^F)_V^T \leq \Delta\sqrt{\frac{A}{t}}|\mathcal{I}_{TR}| + tN_S\epsilon_S$ where $\Delta$ is a span of leaf utilities, $\Delta = \max_{z \in Z} u_i(z) - \min_{z \in Z} u_i(z)$, $A$ is an upper bound on the number of actions, $|\mathcal{I}_{TR}|$ is a number of information sets in the trunk, $N_S$ is the number of information sets at the root of any subgame, and value function error is at most $\epsilon_S$.
\label{lem:convergence}
\end{lemma}
    \section{Safe Model Exploitation}
While CDBR maximizes the exploitation of the fixed opponent model, it allows a player to be exploited. When we face an opponent unsure if our model is perfect we must limit our exploitability. For example, when we gradually build a model during play, we must limit our exploitability in the initial game rounds when the model is still very inaccurate.

\subsection{Combination of CDBR and Nash Equilibrium}
The combination of CDBR and Nash equilibrium (CDBR-NE) is the first approach to limit exploitability. We can simultaneously compute both strategies using depth-limited solving and do a linear combination in every decision node. Let $p$ be the linear combination parameter and $\sigma_\ps^F$ be the opponent model. The gain and exploitability are limited accordingly.

$$
    \sigma^{LC}_\pr =  p\sigma^{NE}_\pr  +  (1 - p) \sbr(\sigma_\ps^F)_V^\mathcal{P}
$$
$$
    \mathcal{E}(\sigma^{LC}_\pr) = p \mathcal{E}(\sigma^{NE}_\pr) + (1 - p)\mathcal{E}(\sbr(\sigma_\ps^F)_V^\mathcal{P})
$$
$$
    \mathcal{G}(\sigma^{LC}_\pr, \sigma_{\ps}^F) =  p\mathcal{G}(\sigma^{LC}_\pr, \sigma_{\ps}^F)  + (1 - p) \mathcal{G}(\sbr(\sigma_\ps^F)_V^\mathcal{P}, \sigma_{\ps}^F)
$$
Desired exploitability or gain may be achieved by tuning the parameter $p$ while being only two times slower than the CDBR since we need to find the Nash equilibrium separately and perform CDBR. The required value function is the same for both parts and is still the same as in standard depth-limited solving.

Required computation is exactly running standard depth-limited solving and CDBR in parallel. Since CDBR computation has standard depth-limited solving as an upper bound, the required computation is at most twice as much as standard depth-limited solving.

\subsection{Continual Depth-limited RNR}
CDBR-NE is safe, but \cite{johanson2008computing} shows we can get a much better trade-off between gain and exploitability using RNR as it recovers the optimal Pareto set of $\epsilon$-safe best responses \cite{mccracken2004safe}. It also gives us better control of safety as it links the allowed exploitability to the achieved gain. We combine depth-limited solving with RNR to create CDRNR.

\subsubsection{Description of Restricted Nash Response}
For CDRNR, we first need to explain the RNR method briefly \cite{johanson2008computing}. RNR is solved by computing a modified game, adding an initial chance node with two outcomes that player $\pr$ does not observe. We copy the whole game tree under both chance node outcomes, and in one tree, the opponent plays the fixed strategy, and we denote it $G^F$. In the other tree, the opponent can play as he wants, resulting in a best response to the strategy of player $\pr$. We denote the other tree $G'$. Since player $\pr$ does not observe the initial chance node, his information sets span over $G'$ and $G^F$, and we denote the full modified game with both trees $G^M$. Parameter $p$ is the method to control the safety and is the initial probability of picking $G^F$.

\subsubsection{Definition}
Given the opponent's fixed strategy $\sigma_\ps^F$ and some subgame partitioning $\mathcal{P}$ of $G^M$, we define continual depth-limited restricted Nash response (CDRNR) recursively from the top. First, we have trunk $T_1^M$ using $\mathcal{P}$ and value function $V$. CDRNR for player $\pr$ in the trunk $T_1^M$ using value function $V$ is $\srnr(\sigma_\ps^F, p)_V^{T_1^M} = \argmax_{\sigma_\pr}u_\pr(\sigma_\pr, BR(\sigma_\pr))_V^{T_1^M}$. And then, in every following step, we create the new subgame $S_i^M$ and enlarge the trunk to incorporate this subgame, creating trunk $T_i^M$ = $T_{i-1}^M \cup S_i^M$. Next, we fix strategy $\sigma_\pr^{T_{i-1}^M}$ of player $\pr$ in the previous trunk $T_{i-1}^M$ and the CDRNR is $\srnr(\sigma_\ps^F, p)_V^{T_i} = \argmax_{\sigma_\pr^{S_i^M}} u_\pr(\sigma_\pr', BR(\sigma_\pr'))_V^{T_i^M}$ where $\sigma_\pr'$ is a combination of the strategy we optimize over and the fixed strategy from the previous step, formally $\sigma_\pr' = \sigma_\pr^{S_i^M} \cup \sigma_\pr^{T_{i-1}^M}$.

To summarize, we optimize only over the strategy in the subgame used in the current step while the strategy in the previous parts of the game is fixed for player $\pr$. The strategy of the opponent is fixed in $G^F$ and free in $G'$. We denote the full CDRNR strategy $\srnr(\sigma_\ps^F, p)_V^\partition$.

\subsubsection{Computing CDRNR}
In practice, we want to avoid duplicating the tree, and we also want to use the exact same value function as in the standard depth-limited solving. We explain why the RNR does not need the duplicated trees in practice. It only needs the reaches of the fixed strategy injected to the terminal nodes in the ratio defined by the parameter $p$. This allows us to precompute the reaches, run CFR as in standard depth-limited solving, and then modify the computed reaches from the iteration using the precomputed fixed reaches. However, we also need to query the value function, which differs from the previous one in the theoretical definition as it spans over the modified public state. However, since the reaches of $p_\pr$ are the same for $G'$ and $G_F$ we can compute it only once by joining the reaches together as in the previous example and querying the standard value function.

So far, we described exactly the standard depth-limited solving with only one modification: modifying the reaches using the fixed strategy. We also use the gadget since now the opponent can deviate in the $G'$. However, standard gadgets will fail due to the addition of imperfect parts of the opponent, and we discuss details along with a solution in the next section.
    \section{Gadgets and Model Exploitation}
\label{sec:gadgets}
When we exploit an opponent model, we need to worsen the strategy in terms of exploitability. We must limit how much the strategy worsens if we want a safe response. Gadgets are used to ensure exploitability does not increase \cite{moravcik2016refining,burch2014solving,brown2017safe}, and all the common gadgets work in scenarios where we do not expect our strategy to worsen. However, we need to worsen our strategy to exploit the opponent. We try to gain as much as possible in RNR in $G^F$. As soon as the strategy gets worse and the exploitability increases, the common gadgets fail to quantify this increase, which is crucial in applications doing a delicate trade-off. The requirement for the gadget which would work in CDRNR is in Definition \ref{def:value}

\begin{definition}
    For each information set $I \in \Iset_{\pr}$ we need the value of its part in $G'$, formally $\sum_{h \in I, h \in G', z \in Z, h \sqsubset z}\pi^\sigma(z)u_{\pr}(z)$, to be the same as the value we would get if we let player $\ps$ play BR in full $G'$.
    \label{def:value}
\end{definition}

The following examples show that the requirement is not satisfied for common resolving gadgets. We tried to construct a gadget that would satisfy the condition, but in the end, we kept the previously resolved parts of the game $G'$ followed by the value function which the game does not follow. We call the construction the full gadget, it satisfies the condition, and still only increases the size of the solved part linearly. Constant-size gadget fulfilling the Definition \ref{def:value} is an open problem.

\subsection{Restricted Nash Response with Gadget}
We show that commonly used resolving gadgets are either overestimating or underestimating the values from Definition~\ref{def:value} on an example game in Figure~\ref{fig:gadget_problem_resolving_gadget}. In the game, we first randomly pick a red or green coin. Player $\ps$ observes this and decides to place the coin heads up (RH, GH) or tails up (RT, GT). Player $\pr$ cannot observe anything and ultimately chooses whether he wants to play the game (P) or quit (Q).

In equilibrium, player $\pr$ plays action $Q$, and player $\ps$ can mix actions up to the point where the utility for $P$ is at most 0. This gives the value of the game 0, and counterfactual values in all inner nodes are also 0. Assuming the modified RNR game $G^M$ with an opponent model playing $GT$, that makes it worth for player $\pr$ to play (P) in the game, $\ps$ will play $(RH,GH)$ in $G'$ with utility -3 for player $\pr$ in $G'$. We will use gadgets to resolve the game from the player $\pr$ information set.

\def\istgamescale{0.8}
\begin{figure*}
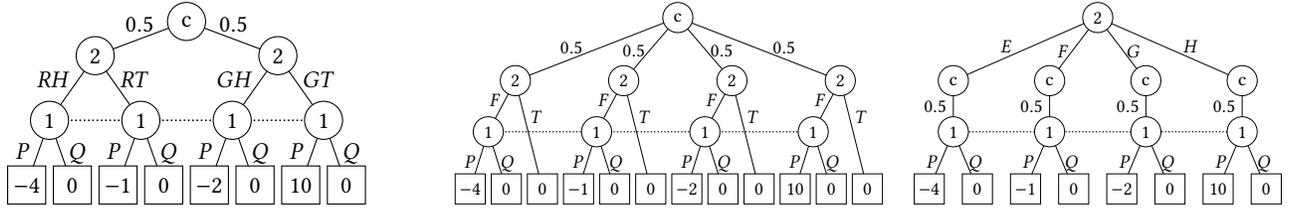

    \begin{minipage}{0.32\linewidth}
        \centering
        \scalebox{1.012}{
            \begin{istgame}
            \setistOvalNodeStyle{.5cm}
            \setistRectangleNodeStyle{.5cm}
            \xtShowTerminalNodes[rectangle node]
            \istrooto(0){c}+{4.5mm}..{2.4cm}+
                \istb{0.5}[left, xshift=4mm, yshift=2.2mm]  \istb{0.5}[right, xshift=-4mm, yshift=2.2mm]  \endist
            \istrooto(1)(0-1){$\ps$}+{8.5mm}..{1.2cm}+
                \istb{RH}[left, xshift=1.5mm, yshift=2mm]  \istb{RT}[right, xshift=-1.5mm, yshift=2mm]  \endist
            \istrooto(2)(0-2){$\ps$}+{8.5mm}..{1.2cm}+
                \istb{GH}[left, xshift=1.5mm, yshift=2mm]  \istb{GT}[right, xshift=-1.5mm, yshift=2mm]  \endist
            \istrooto(3)(1-1){$\pr$}+{8.5mm}..{.6cm}+
                \istbt{P}[l]{-4}[center]  \istbt{Q}[r]{0}[center]  \endist
            \istrooto(4)(1-2){$\pr$}+{8.5mm}..{.6cm}+
                \istbt{P}[l]{-1}[center]  \istbt{Q}[r]{0}[center]  \endist
            \istrooto(5)(2-1){$\pr$}+{8.5mm}..{.6cm}+
                \istbt{P}[l]{-2}[center]  \istbt{Q}[r]{0}[center]  \endist
            \istrooto(6)(2-2){$\pr$}+{8.5mm}..{.6cm}+
                \istbt{P}[l]{10}[center]  \istbt{Q}[r]{0}[center]  \endist
            \xtInfoset(1-1)(1-2)
            \xtInfoset(1-2)(2-1)
            \xtInfoset(2-1)(2-2)
            \end{istgame}
        }
    \end{minipage}
    \hfill
    \begin{minipage}{0.32\linewidth}
        \centering
        \scalebox{\istgamescale}{
            \begin{istgame}
            \setistOvalNodeStyle{.5cm}
            \setistRectangleNodeStyle{.5cm}
            \xtShowTerminalNodes[rectangle node]
            \istrooto(0){c}+{10.5mm}..{1.8cm}+
                \istb{0.5}[left, xshift=0mm, yshift=0.5mm]  \istb{0.5}[left, xshift=1mm, yshift=0.5mm] \istb{0.5}[right, xshift=-1mm, yshift=0.5mm] \istb{0.5}[right, xshift=0mm, yshift=0.5mm]  \endist
            \istrooto(1)(0-1){$\ps$}+{8.5mm}..{0.9cm}+
                \istb{F}[left, xshift=1.5mm, yshift=2mm]  \istbAt<level distance=1.2*\xtlevdist>{T}[right, xshift=-0.8mm, yshift=2.8mm]{0}[center]  \endist
            \istrooto(2)(0-2){$\ps$}+{8.5mm}..{0.9cm}+
                \istb{F}[left, xshift=1.5mm, yshift=2mm]  \istbAt<level distance=1.2*\xtlevdist>{T}[right, xshift=-0.8mm, yshift=2.8mm]{0}[center]  \endist
            \istrooto(3)(0-3){$\ps$}+{8.5mm}..{0.9cm}+
                \istb{F}[left, xshift=1.5mm, yshift=2mm]  \istbAt<level distance=1.2*\xtlevdist>{T}[right, xshift=-0.8mm, yshift=2.8mm]{0}[center]  \endist
            \istrooto(4)(0-4){$\ps$}+{8.5mm}..{0.9cm}+
                \istb{F}[left, xshift=1.5mm, yshift=2mm]  \istbAt<level distance=1.2*\xtlevdist>{T}[right, xshift=-0.8mm, yshift=2.8mm]{0}[center]  \endist
            \istrooto(5)(1-1){$\pr$}+{9.5mm}..{.6cm}+
                \istbt{P}[left, yshift=-0.5mm, xshift=0.5mm]{-4}[center]  \istbt{Q}[right, yshift=-0.5mm, xshift=-0.5mm]{0}[center]  \endist
            \istrooto(6)(2-1){$\pr$}+{9.5mm}..{.6cm}+
                \istbt{P}[left, yshift=-0.5mm, xshift=0.5mm]{-1}[center]  \istbt{Q}[right, yshift=-0.5mm, xshift=-0.5mm]{0}[center]  \endist
            \istrooto(7)(3-1){$\pr$}+{9.5mm}..{.6cm}+
                \istbt{P}[left, yshift=-0.5mm, xshift=0.5mm]{-2}[center]  \istbt{Q}[right, yshift=-0.5mm, xshift=-0.5mm]{0}[center]  \endist
            \istrooto(8)(4-1){$\pr$}+{9.5mm}..{.6cm}+
                \istbt{P}[left, yshift=-0.5mm, xshift=0.5mm]{10}[center]  \istbt{Q}[right, yshift=-0.5mm, xshift=-0.5mm]{0}[center]  \endist
            \xtInfoset(1-1)(2-1)
            \xtInfoset(2-1)(3-1)
            \xtInfoset(3-1)(4-1)
            \end{istgame}
        }
    \end{minipage}
    \begin{minipage}{0.32\linewidth}
        \centering
        \scalebox{\istgamescale}{
            \begin{istgame}
            \setistOvalNodeStyle{.5cm}
            \setistRectangleNodeStyle{.5cm}
            \xtShowTerminalNodes[rectangle node]
            \istrooto(0){$\ps$}+{10.5mm}..{1.6cm}+
                \istb{E}[left, xshift=0mm, yshift=0.8mm]  \istb{F}[left, xshift=1mm, yshift=0.5mm] \istb{G}[right, xshift=-1mm, yshift=0.5mm] \istb{H}[right, xshift=0mm, yshift=0.8mm]  \endist
            \istrooto(1)(0-1){c}+{8.5mm}..{0.8cm}+
                \istb{0.5}[left, xshift=0mm, yshift=1mm] \endist
            \istrooto(2)(0-2){c}+{8.5mm}..{0.8cm}+
                \istb{0.5}[left, xshift=0mm, yshift=1mm] \endist
            \istrooto(3)(0-3){c}+{8.5mm}..{0.8cm}+
                \istb{0.5}[left, xshift=0mm, yshift=1mm] \endist
            \istrooto(4)(0-4){c}+{8.5mm}..{0.8cm}+
                \istb{0.5}[left, xshift=0mm, yshift=1mm] \endist
            \istrooto(5)(1-1){$\pr$}+{9.5mm}..{.8cm}+
                \istbt{P}[left, yshift=-0.5mm, xshift=0.5mm]{-4}[center]  \istbt{Q}[right, yshift=-0.5mm, xshift=-0.5mm]{0}[center]  \endist
            \istrooto(6)(2-1){$\pr$}+{9.5mm}..{.8cm}+
                \istbt{P}[left, yshift=-0.5mm, xshift=0.5mm]{-1}[center]  \istbt{Q}[right, yshift=-0.5mm, xshift=-0.5mm]{0}[center]  \endist
            \istrooto(7)(3-1){$\pr$}+{9.5mm}..{.8cm}+
                \istbt{P}[left, yshift=-0.5mm, xshift=0.5mm]{-2}[center]  \istbt{Q}[right, yshift=-0.5mm, xshift=-0.5mm]{0}[center]  \endist
            \istrooto(8)(4-1){$\pr$}+{9.5mm}..{.8cm}+
                \istbt{P}[left, yshift=-0.5mm, xshift=0.5mm]{10}[center]  \istbt{Q}[right, yshift=-0.5mm, xshift=-0.5mm]{0}[center]  \endist
            \xtInfoset(8)(7)
            \xtInfoset(7)(6)
            \xtInfoset(6)(5)
            \end{istgame}
        }
    \end{minipage}
    \caption{(left) A game to show problems with gadgets. (middle) Resolving gadget for the left game. (right) Max-margin and Reach max-margin gadget.}
    \label{fig:gadget_problem_resolving_gadget}
\end{figure*}
\paragraph{Resolving Gadget} \cite{burch2014solving}
Resolving gadget constructs a game that allows the opponent to choose whether he wants to play in the subgame we created or terminate. It is done by inserting nodes above the roots of the subgame, and the opponent has two actions before each root, either to follow and play the game or to terminate and receive a reward they would get by playing the previously resolved equilibrium. Those nodes are grouped into information sets based on the opponent's augmented information sets at the subgame's roots.

Resolving gadget on the game in Figure~\ref{fig:gadget_problem_resolving_gadget} has all utilities after \textit{terminate} actions 0. When we resolve the gadget, the utility is 0. However, when player $\pr$ deviates to action $P$, player $\ps$ plays \textit{follow} action in all but the rightmost node, and the utility of player $\pr$ will be -3.5. Therefore, the common resolving gadget may overestimate the real exploitability of the strategy in the subgame. Overestimating may lead to not exploiting as much as we can and makes it impossible to prove Theorem~\ref{thm:SRNR nash} about the minimal gain of our algorithm. Normalization of the chance node might seem to solve the problem, but it would only halve the value to -1.75, which is still incorrect.

\paragraph{(Reach) Max-margin Gadget} \cite{moravcik2016refining,brown2017safe}
Both reach max-margin and max-margin gadgets allow the opponent to choose any information set at the start of the subgame. This is done by inserting a single node to the top, where the opponent has an action for each of his augmented information sets in the root of the subgame. After the action is a chance node to split the information set to the histories, with correct reaches by the resolving player and chance. Furthermore, all the terminal values are adjusted by the same value, which is in the terminate action in the resolving gadget. In the reach max-margin gadget, this value is further modified by an approximation of opponent mistakes.

All the counterfactual best response values are 0, and we assume both players played perfectly before the depth limit. Hence, we do not need to offset any node in the (reach) max-margin gadget. Then, both gadget constructions are identical. We add the initial decision node and the chance nodes (since there is only one state in each information set, the nodes have only one action). When we solve the gadget, player $\pr$ will pick action $Q$, and the gadget value will be 0. However, when player $\pr$ deviates to action $P$, player $\ps$ now has a choice between terminal utilities and picks action $E$ to receive the highest one. This will result in utility -2, and we see that (reach) max-margin gadgets can underestimate the real exploitability. It can lead to our algorithm being more exploitable than we want using some $p$, and it makes Theorem~\ref{thm:DBR nash} impossible to prove. Similar to the previous gadget, normalizing the chance nodes would lead to double the utility, which is still incorrect.

\paragraph{Full Gadget}
The only construction fulfilling the requirements we found is to keep all the previously explored parts of the game in a path to the root and use a value function when we leave. Using the optimal value function, the construction simulates the best response, which measures exploitability.

Formally, when we reach subgame $S_i$ we construct a composite game by joining $S_i$, the trunk $T$, and all the previous subgames $S_j, j \in {1,...,i - 1}$. It corresponds to the illustration in Figure~\ref{fig:subgames}, and the value function will evaluate every public state $PS$ from which the actions lead outside of the game.

\begin{figure}
    \centering
    \includegraphics[width=1\linewidth]{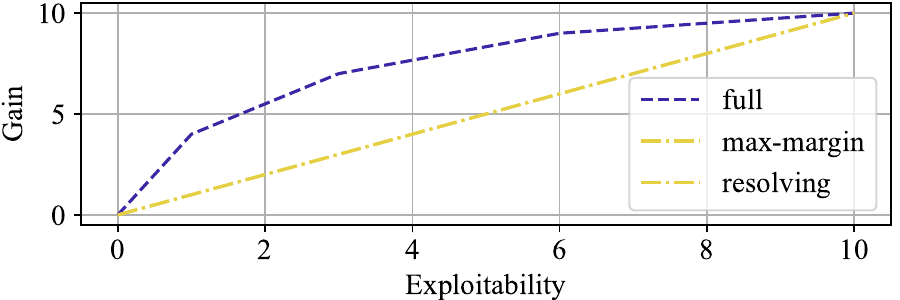}    
    \caption{Comparison of Gain and Exploitability of the solutions using Full gadget compared to other gadgets. Max-margin and resolving gadget lines are overlapping.}
    \label{fig:g_break}
\end{figure}

We show that other gadgets can overestimate or underestimate exploitability, which could shift the distribution of the parameter $p$, and we could still compute the same solutions. However, in Figure~\ref{fig:g_break}, we show the results of games created to break the other gadgets. We have two games in which the full gadget behaves correctly, and each breaks the other gadget. Figure~\ref{fig:g_break} shows the results joined together. Both games have five actions for the exploiter, and each one is crucial in reconstructing the full RNR set. The full gadget can recover all five actions using different $p$, but other gadgets can only compute two actions regardless of the choice of $p$.

\paragraph{Complexity of CDRNR}
We can use other gadgets in CDRNR to obtain fast algorithms without any theoretical guarantees and with the same bound on computation as we have for the combination of CDBR and Nash equilibrium. The soundness of the algorithm relies on using the full gadget, which requires solving increasingly larger parts of the game as the depth increases. This increase in size is linear with the resolving steps, so the full algorithm complexity is quadratic in the depth of the game compared with vanilla continual resolving or CDBR. It makes the algorithm applicable to shorter games like Poker or Goofspiel but infeasible for long games like Stratego. 

\subsubsection{Soundness of CDRNR}
For the following theorems, we denote $\mathcal{S}$ as the set consisting of the trunk and all the subgames explored when computing the response, and $\mathcal{S'}$ is the same but without the last subgame. $S^B$ denotes a border of the subgame. We also denote $S^O$, the set of all the states where we leave the trunk not going in the currently resolving subgame.

\begin{theorem}[Gain of CDRNR]
Let $G$ be any zero-sum extensive-form game and let $\sigma_\ps^\text{F}$ be any fixed opponent's strategy in $G$. Then we set $G^M$ as restricted Nash response modification of $G$ using $\sigma_\ps^\text{F}$. Let $\mathcal{P}$ be any subgame partitioning of the game $G^M$ and using some $p \in \langle0,1\rangle$, let $\sigma_\pr^{\srnr}$ be a CDRNR given approximation $\bar{V}$ of value function $V$ with error at most $\epsilon_V$ and opponent strategy $\sigma_\ps^\text{F}$ approximated in each step with regret at most $\epsilon_R$, formally $\sigma_\pr^{\srnr} = \srnr(\sigma_\ps^F, p)_V^\mathcal{P}$. Let $\sigma^{NE}$ be any Nash equilibrium in $G$. Then $u_\pr(\sigma_\pr^{\srnr}, \sigma_\ps^F) + \sum_{S \in \mathcal{S'}}|I_{S^O}|(1-p)\epsilon_V + |\mathcal{S}|\epsilon_R + \sum_{S \in \mathcal{S'}}|I_{S^B}|\epsilon_V \geq u_\pr(\sigma^{NE})$.
\label{thm:SRNR nash}
\end{theorem}

The previous theorem states that our approaches will receive at least the value of the game when responding to the model. All the proofs are in the appendix.

\begin{theorem}[Safety of CDRNR]
Let $G$ be any zero-sum extensive-form game and let $\sigma_\ps^\text{F}$ be any fixed opponent's strategy in $G$. Then we set $G^M$ as restricted Nash response modification of $G$ using $\sigma_\ps^\text{F}$. Let $\mathcal{P}$ be any subgame partitioning of the game $G^M$ and using some $p \in \langle0,1\rangle$, let $\sigma_\pr^{\srnr}$ be a CDRNR given approximation $\bar{V}$ of the optimal value function $V$ with error at most $\epsilon_V$, partitioning $\mathcal{P}$ and opponent strategy $\sigma_\ps^\text{F}$, which is approximated in each step with regret at most $\epsilon_R$, formally $\sigma_\pr^{\srnr} = \sigma_\pr^\srnr(\sigma_\ps^F, p)_V^\mathcal{P}$. Then exploitability has a bound $\mathcal{E}(\sigma_\pr^\srnr) \leq \mathcal{G}(\sigma_\pr^\srnr, \sigma_\ps^F)\frac{p}{1-p} + \sum_{S \in \mathcal{S'}}|I_{S^O}|(1-p)\epsilon_V + |\mathcal{S}|\epsilon_R + \sum_{S \in \mathcal{S'}}|I_{S^B}|\epsilon_V$, $\mathcal{E}$ and $\mathcal{G}$ are defined in Section~\ref{sec:background}.
\label{thm:DBR nash}
\end{theorem}

The last theorem is more complex, and it bounds the exploitability by the gain of the strategy against the model. With $p = 0$, it is reduced to the continual resolving, and with $p = 1$ to CDBR with unbounded exploitability. The theorem shows the parameter $p$ directly links allowed exploitability to the gain we receive. The same works in RNR without the resolving and value errors; as far as we know, the authors do not explicitly mention it.

SES has bound relies on opponent estimation being close to an equilibrium strategy. When the estimation is more different, the bound is infinity for a large portion of the parameter alpha. We give a detailed explanation in the appendix.

More intuitively, Theorem~\ref{thm:SRNR nash} says that by playing the proposed algorithm, we have at least the same safety guarantees we would get by playing a NE against the opponent we modeled correctly. Theorem~\ref{thm:DBR nash} allows us to choose a trade-off between the exploitation of the opponent behaving according to the model and safety against an opponent who would deviate arbitrarily from the model.
    \section{Experiments}
We compared CDBR and local best response (LBR) \cite{lisy2017eqilibrium}. We empirically show the performance of CDRNR and explore the trade-off between exploitability and gain in CDRNR. The appendix contains hardware setup, domain description, algorithm details, and experiments on more domains. We use two types of opponent strategies: strategies generated by few CFR iterations and random strategies with different seeds.

\subsection{SES explanation}
Safe exploitation search (SES) \cite{liu2022safe} is a similar method to the one we propose. However, there are two significant differences. First, the method uses a max-margin gadget without the analysis we did. Hence, the bound of exploitability is very loose, and for a wide range of inputs, the exploitability can be unbounded. Second, the method does not fix the opponent's strategy at all and only uses opponent reaches when resolving the subgame. As a result, SES exploitation is very limited, and as we show in experiments, it is often worse than using the best Nash equilibrium. On top of that, in some games, it fundamentally cannot exploit the opponent, notably in any perfect information game, even with simultaneous moves.

\begin{figure}[ht]
    \centering
    \includegraphics[width=1\linewidth]{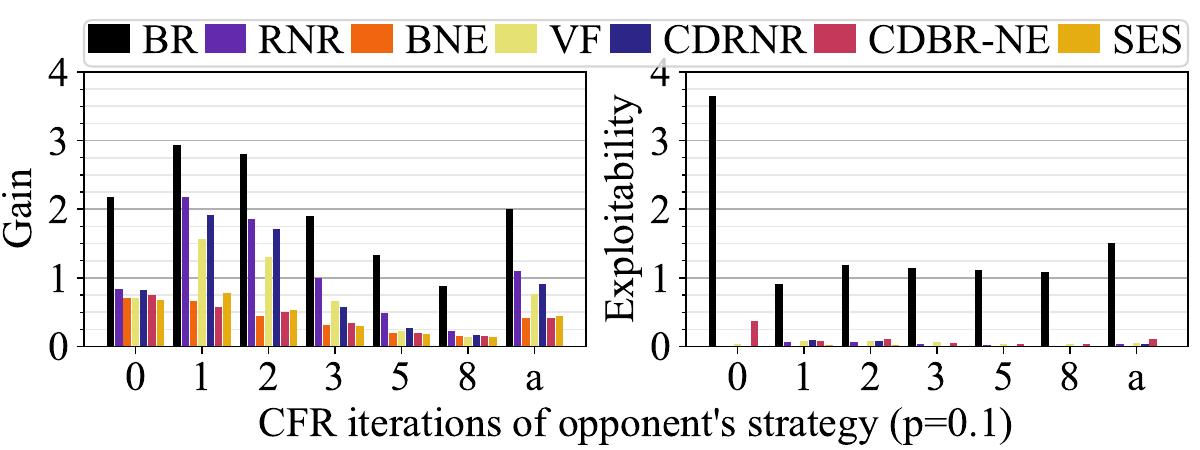}
    \includegraphics[width=1\linewidth]{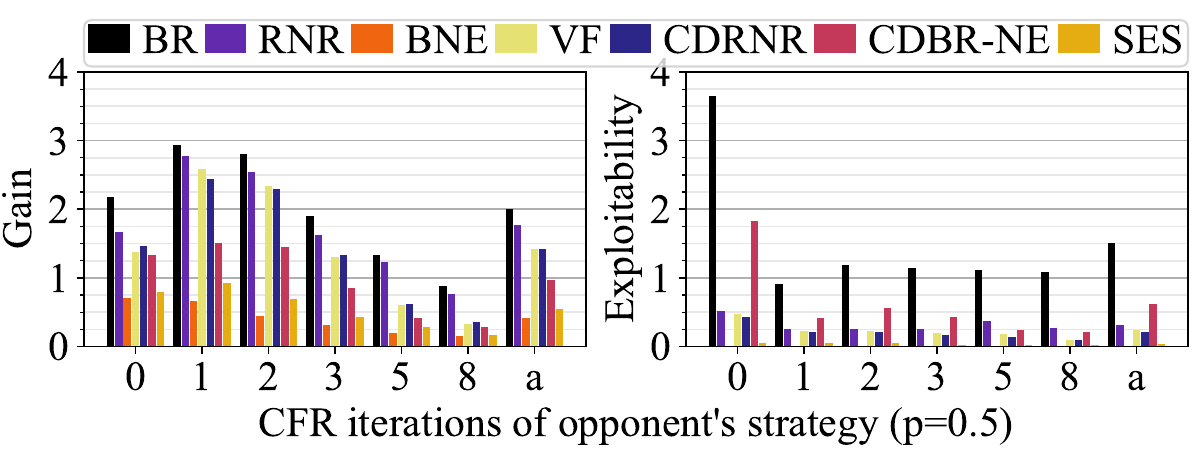}
    \includegraphics[width=1\linewidth]{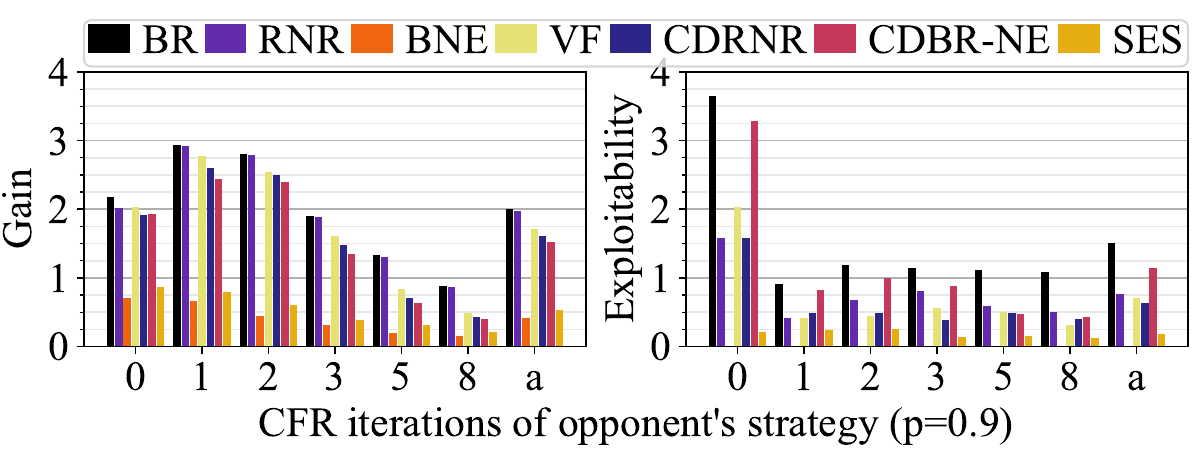}
    \caption{Gain and exploitability comparison of BR, RNR, best Nash equilibrium (BNE), CDBR-NE, SES, and CDRNR in Leduc Hold'em against strategies from CFR using a small number of iterations with different $p$ values. The \textbf{a} stands for the average of the other values. VF is CDRNR using an imperfect value function.}
    \label{fig:expl}
\end{figure}

\subsection{Exploitability of Robust Responses}
We report both gain and exploitability for CDRNR on Leduc Hold'em. Results in Figure~\ref{fig:expl} show that the proven bound on exploitability works in practice, and we see that the bound is very loose in practice. For example, with $p=0.5$, the bound on the exploitability is the gain itself, but the algorithm rarely reaches even a tenth of the gain in exploitability. This shows that the CDRNR is similar to the restricted Nash response because, with a well-set $p$, it can significantly exploit the opponent without significantly raising its exploitability. The only exception is the CDRNR with a value function, which shows the added constants from the value function's imprecision. When the opponent is close to optimal, we see that the exploitability can rise above the gain.

In most cases, the gain and exploitability of CDRNR are lower than that of RNR. Gain must be lower because of the different value function, but the exploitability can be higher, as seen with $p=0.1$ against low iteration strategies, due to the depth-limited nature. We provide results on other domains and for more parameter values $p$ in the appendix.

We also compare the algorithm against the best possible Nash equilibrium. We compute the best NE by a  linear program, and it serves as the theoretical limit of maximal gain, which does not allow exploitability. It would be impossible to compute for larger games. We can see we can gain more than twice as much, with exploitability still being almost zero. 

In our results on smaller games, we use the optimal value function computed by a linear program. To show the performance of imperfect value function we included results where the value function is approximated by CFR with 500 iterations. As expected, the imperfect value function slightly decreases the performance but is comparable to the algorithms using the optimal value function.

The last comparison is with SES, which performs poorly, and its gain is only slightly above the best Nash equilibrium. Conversely, it is almost not exploitable. Our results are consistent with results in the paper \cite{liu2022safe} and are a direct consequence of using the information from the opponent model only to set the reaches to the subgame. Only reaches are not enough to do meaningful exploitation, and SES produces strategies that are very close to Nash equilibrium.

\subsection{ABR vs CDBR}
We compared CDBR with ABR \cite{timbers2020approximate} on Leduc and different imperfect information Goofspiel. The results are in Table~\ref{tbl:abr_vs_cdbr}, showing that our method is slightly behind in Leduc, even with the highest search depth. In Goofspiel, CDBR1, which looks only one action into the future, is already pretty good, and as soon as we allow CDBR to look three turns in the future, it fully exploits the opponents. In IIGS4, that is half of the game, but in IIGS6, it is less than $\frac{1}{3}$ of the game, and it is still enough.

\begin{table}
    \centering
        \setlength\tabcolsep{6pt}
        \renewcommand{\arraystretch}{1.2}
        \begin{tabular}{|c|c|c|c|c|c|}
            \hline
            & \textbf{ABR} & \textbf{CDBR1} & \textbf{CDBR3} & \textbf{CDBR5}  & \textbf{BNE}\\ \hline
            Leduc & 98\% & 74.4\% & 96\% & 97.1\% & 32.6\% \\  \hline
            IIGS4 & 97\% & 98.5\% & 100\% & 100\% & 77.1\% \\  \hline
            IIGS5 & 95\% & 93.5\% & 100\% & 100\% & 50.4\% \\  \hline
            IIGS6 & 97\% & 90.7\% & 100\% & 100\% & 45.9\% \\  \hline
        \end{tabular}    
    \caption{Comparison of ABR and CDBR with BNE baseline on different games against uniform random. The values are the percentage of gain achieved by the best response.}
    \label{tbl:abr_vs_cdbr}
\end{table}

\subsection{Local Best Response vs. CDBR}
We compare LBR and CDBR in Leduc Hold'em. We also compare CDBR with the BR in imperfect information Goofspiel 5, but without LBR, which is poker-specific. We show that CDBR and LBR are very similar with smaller steps, and as we increase the depth limit of CDBR, it starts outperforming LBR. The behavior differs slightly based on the specific strategy because LBR assumes the player continues the game by only calling till the end, while CDBR uses the perfectly rational extension.

The results in Figure~\ref{fig:lbrsbr} show that both concepts are good at approximating the best response, with CDBR being better against both strategies. LBR looks at one following action, so it is best compared to the CDBR1 in terms of comparability. Next, we observe a lack of monotonicity in step increase, which is explained with an example in the appendix. When we increase the depth limit, the algorithm can exploit some early mistake that causes it to miss a sub-tree where the opponent makes a much bigger mistake in the future. We can clearly see the difference between the algorithm with guarantees and LBR without them. Against strategy from 34 CFR iterations, LBR can no longer achieve positive gain and only worsens with more iterations. In contrast, CDBR can always achieve at least zero gain, assuming we have an optimal value function.

\begin{figure}
    \centering
    \includegraphics[width=1\linewidth]{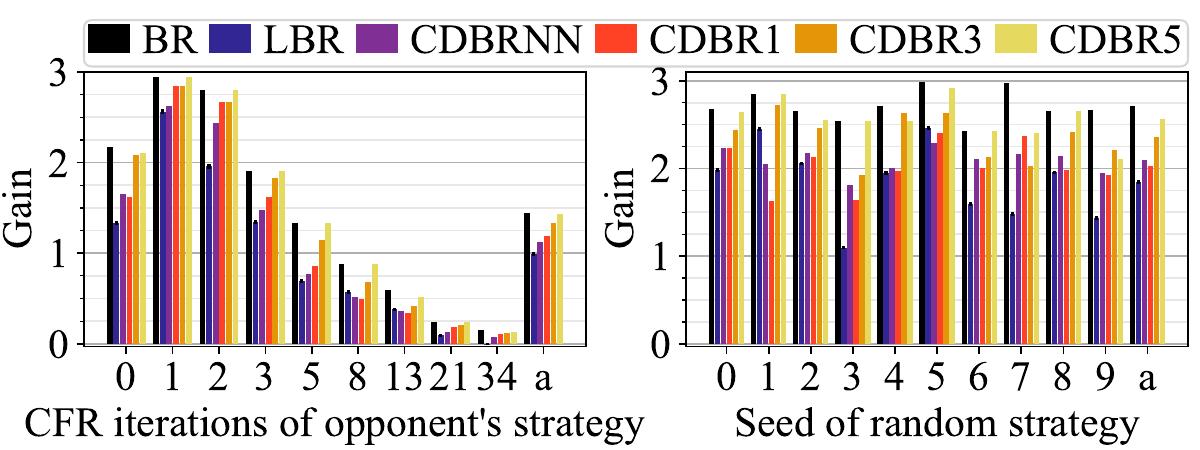}
    \includegraphics[width=1\linewidth]{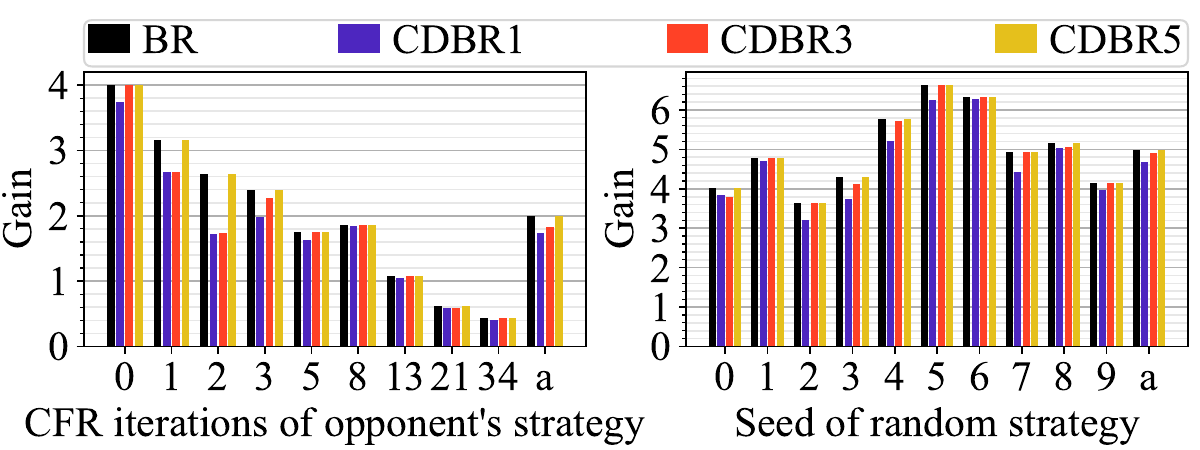}
    \caption{Gain comparison of best response (BR), local best response (LBR - only poker), and continual depth-limited best response (CDBR) in Leduc Hold'em (top) and IIGoofspiel 5 (bottom) against strategies from CFR using a small number of iterations (left) and random strategies (right). The \textbf{a} stands for the average of the other values in the plot. The number after CDBR stands for the number of actions CDBR was allowed to look at in the future, and CDBRNN is a one-step CDBR with a neural network as a value function.}
    \label{fig:lbrsbr}
\end{figure}

\subsection{Playing SlumBot}
We tested our method in HUNL against SlumBot \cite{jackson2017targeted}, which is a publicly available abstraction-based bot commonly used for benchmarking. We used a fold, call, pot, and all-in (FCPA) abstraction for CDBR, and we also rerun the results of LBR against the new SlumBot (since the author of SlumBot confirmed that the strategy we were provided is different from the one LBR used but the same as the one provided to authors of ABR). CDBR significantly outperforms both ABR and LBR and we report the results in Table~\ref{tbl:slumbot_results}. We tried to run the LBR restricted to only call in the first two rounds but found it no longer helps against the new SlumBot, and we reported results for LBR with FCPA. Authors in \cite{timbers2020approximate} also use FCPA for their method but did not report a confidence interval for the results. However, the difference is large enough to have statistical significance if we assume they played over 50 thousand hands.

\begin{table}
    \centering
        \setlength\tabcolsep{6pt}
        \renewcommand{\arraystretch}{1.2}
        \begin{tabular}{|c|c|c|c|c|c|c|c|c|}
            \hline
            & \textbf{ABR} & \textbf{LBR} & \textbf{CDBR} \\ \hline
            Win-rate [mbb/h] & 1259 ± ? & 1388 ± 150 & 1774 ± 137 \\  \hline
        \end{tabular}
    \caption{Comparison of CDBR with LBR and ABR against SlumBot. Results are reported in milibigblinds per hand (mbb/h) with 95\% confidence intervals. (Authors of ABR did not provide confidence intervals.)}
    \label{tbl:slumbot_results}
\end{table}
    \section{Related Work}
This section describes the related work focusing more on distinguishing our novel contributions.

Restricted Nash response (RNR) \cite{johanson2008computing} is an opponent-exploiting scheme. It solves the entire game and allows changing the trade-off between exploitability and gain. Essentially it always produces $\epsilon$ safe best response \cite{mccracken2004safe}. It accomplishes the goal by copying the whole game and then fixing the opponent in one part while having a chance node at the top decide which game we play.  

However, it is impossible to compute RNR in huge games, and we fused the RNR approach with depth-limited solving creating a novel algorithm we call CDRNR. CDRNR is the best performing theoretically sound robust response calculation that can be done in huge games, enabling new opponent exploiting approaches.

Local best response \cite{lisy2017eqilibrium} is an evaluation tool for poker. It uses a given abstraction in its action space. It picks the best action in each decision set, looking at the fold probability for the opponent in the next decision node and then assuming the game is called until the end. Our algorithm CDBR is a generalization of the LBR because we can use it on any game solvable by depth-limited solving. In the algorithm, we have explicitly defined value function, which we can exchange for different heuristics.

Approximate best response (ABR) \cite{timbers2020approximate} is also a generalization of the LBR and showed promising results in evaluating strategies. However, our approach focuses on model exploitation, which requires crucial differences, such as quick re-computation against unseen models. ABR needs to independently learn the response for every combination of opponent and game, making it unusable in the opponent modeling scenario. Our algorithms learn a single domain-specific value function and can subsequently compute strategies against any opponent in the run-time. Furthermore, ABR and even CDBR are extremely brittle, making it a bad choice if we are unsure about the opponent, which we often are in a game against an unknown opponent. On the other hand, CDRNR tackles exactly this issue and provides powerful exploitation with very limited exploitability.

Another reinforcement-learning (RL) method uses neuroevolution with RL, they show a significant increase in performance over DQN and evaluate their method on HUNL. However, they do not share the details of the baseline opponents they played against. We tried to contact the authors without any response and we could not compare the performance with our methods. \cite{xu2021efficient}
    \section{Conclusion}
Opponent modeling and exploitation is an essential topic in computational game theory, with many approaches attempting to model and exploit opponents in various games. However, exploiting opponents in very large games is not trivial, and only recently was an algorithm created to exploit models in depth-limited solving. We explain the problem arising from the inability of gadgets to measure exploitability and we propose a full gadget that solves the issue. We propose a new algorithm to quickly compute depth-limited best response and depth-limited restricted Nash response once we have a value function, creating the best performing theoretically sound robust response applicable to large games. Finally, we empirically evaluate the algorithms on multiple games. We show that CDBR outperforms LBR in both Leduc and HUNL and we show that CDBR performs significantly better against SlumBot than any other previous method. Finally, we show that CDRNR outperforms SES in any game and can achieve over half of the possible gain without almost any exploitability.
    
    \if\printacknowledgments1
    \section*{Acknowledgements}
        This research was supported by Czech Science Foundation (grant no. GA22-26655S) and and the Grant Agency of the Czech Technical University in Prague, grant No. SGS22/168/OHK3/3T/13. Computational resources were supplied by the project "e-Infrastruktura CZ" (e-INFRA CZ LM2018140) supported by the Ministry of Education, Youth and Sports of the Czech Republic and also the OP VVV funded project CZ.02.1.01/0.0/0.0/16\_019/0000765 ``Research Center for Informatics'' which are both gratefully acknowledged. We also greatly appreciate the help of Eric Jackson, who provided the SlumBot strategy and helped with the experiments. 
    \fi

    \bibliographystyle{ACM-Reference-Format} 
    \bibliography{bibliography}   
\fi

\if\printappendix1    
    \clearpage
    \appendix
\section{Pseudocode}
In this section, we present the pseudocode of both CDBR and CDRNR.
\begin{algorithm}
\caption{Computing CDRNR (CDBR)}\label{alg:cdbr}
\begin{algorithmic}
\Require game $G$, model strategy $\sigma_\ps^F$, value function $V$
\State create (virtually) modified game $G^M$ (only CDRNR)
\State create subgame partitioning $\partition$ from $G^M$ ($G$) 
\State $\sigma_\pr^\sbr$ = empty strategy ready to be filled
\State $I$ = initial information set in which we act
\State $S$ = $\empty$ current constructed subgame
\While{$I$ not null}
\If{$I$ not in $S$}
    \State S = construct new $S$ from $\partition$ using previous $S$ (CDRNR does not delete trunk)
    \State $\sigma_\pr^\sbr$ += solution of $S$ using CFR+ with $V$
\Else
    \State pick action $A$ according to $\sigma_\pr^\sbr$ in $I$
    \State get new $I$ using $A$ (or null if the game ends)
\EndIf
\EndWhile
\end{algorithmic}
\end{algorithm}

\section{Additional CDBR results}
We compare LBR and CDBR in Leduc Holde'm. We also compare CDBR with just the BR in imperfect information Goofspiel, but without LBR, which is poker specific. We show that CDBR and LBR are very similar with smaller steps, and as we increase the depth-limit of CDBR, it starts outperforming LBR. The behavior differs in every strategy because LBR assumes the player continues the game by only calling till the end, while CDBR uses the perfectly rational extension. Furthermore, it is possible to exchange the value function of CDBR, and both concepts would be very similar. However, we would lose the guarantee that CDBR will never perform worse than the value of the game.

\begin{figure}
    \centering
    \includegraphics[width=1\linewidth]{graphics/lbr_sbr.pdf}
    \includegraphics[width=1\linewidth]{graphics/gain_iigs5.pdf}
    \includegraphics[width=1\linewidth]{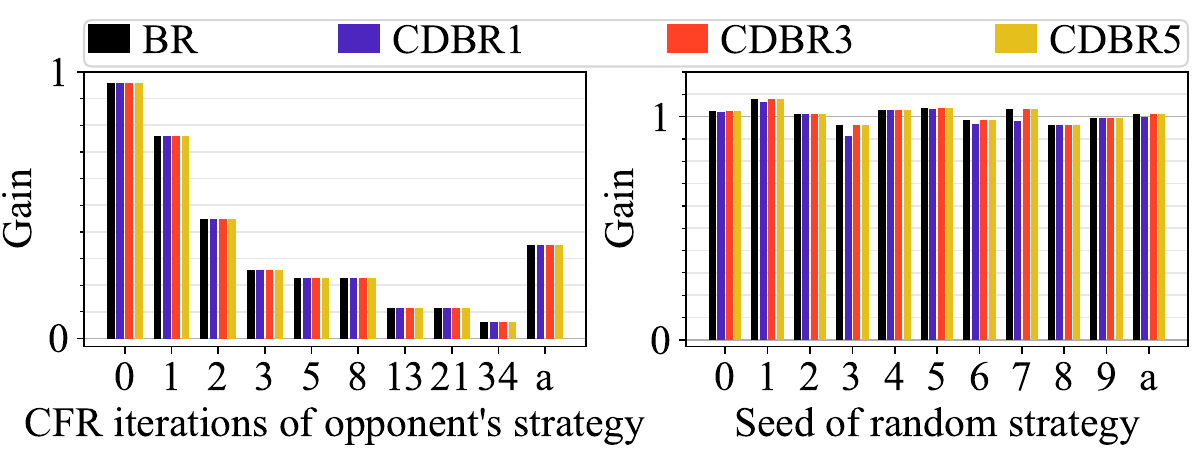}
    \caption{Gain comparison of best response (BR), local best response (LBR - only poker), and continual depth-limited best response (CDBR) in Leduc Hold'em (top), Imperfect information Goofspiel (middle) and Small Liar's dice (bottom) against strategies from CFR using a small number of iterations (left) and random strategies (right). The \textbf{a} stands for the average of the other values in the plot. The number after CDBR stands for the number of actions CDBR was allowed to look in the future, and CDBRNN is one step CDBR with a neural network as a value function.}
    \label{fig:lbrsbrext}
\end{figure}

Looking at the results in Figure~\ref{fig:lbrsbrext}, we can see that both concepts are good at approximating the best response, with CDBR being better against both strategies. LBR looks at one following action, so in terms of comparability, it is best compared to the CDBR1. Next, we observe a lack of monotonicity in step increase, which is linked to the counterexample in Figure \ref{fig:sbr_nash_counterexample}. When we increase the depth-limit, the algorithm can exploit some early mistake that causes it to miss a sub-tree where the opponent makes a much bigger mistake in the future. We can clearly see the difference between the algorithm with guarantees and LBR without them. Against strategy from 34 CFR iterations, LBR can no longer achieve positive gain and only worsens with more iterations. In contrast, CDBR can always achieve at least zero gain (assuming we have an optimal value function).

\setcounter{theorem}{0}
\section{Counterexample Gadget Game}
Examples in Figure~\ref{fig:example_game} are the games used to generate the Figure~\ref{fig:g_break}. The plot in the figure combines two games that have pure actions with the same gain and exploitability. The full gadget reconstructs the Pareto set using all the actions in both games. Other gadgets fail in one of the presented games in Figure~\ref{fig:example_game}. In Figure~\ref{fig:bestactions}, we show the expected utility of all the actions in the CDRNR version of the game, showing that for the resolving gadget and max-margin gadget, two actions dominate all the others, and we can not select any $p$ which would resolve any of the remaining actions.

\begin{figure*}
    \centering
    \includegraphics[width=0.45\textwidth]{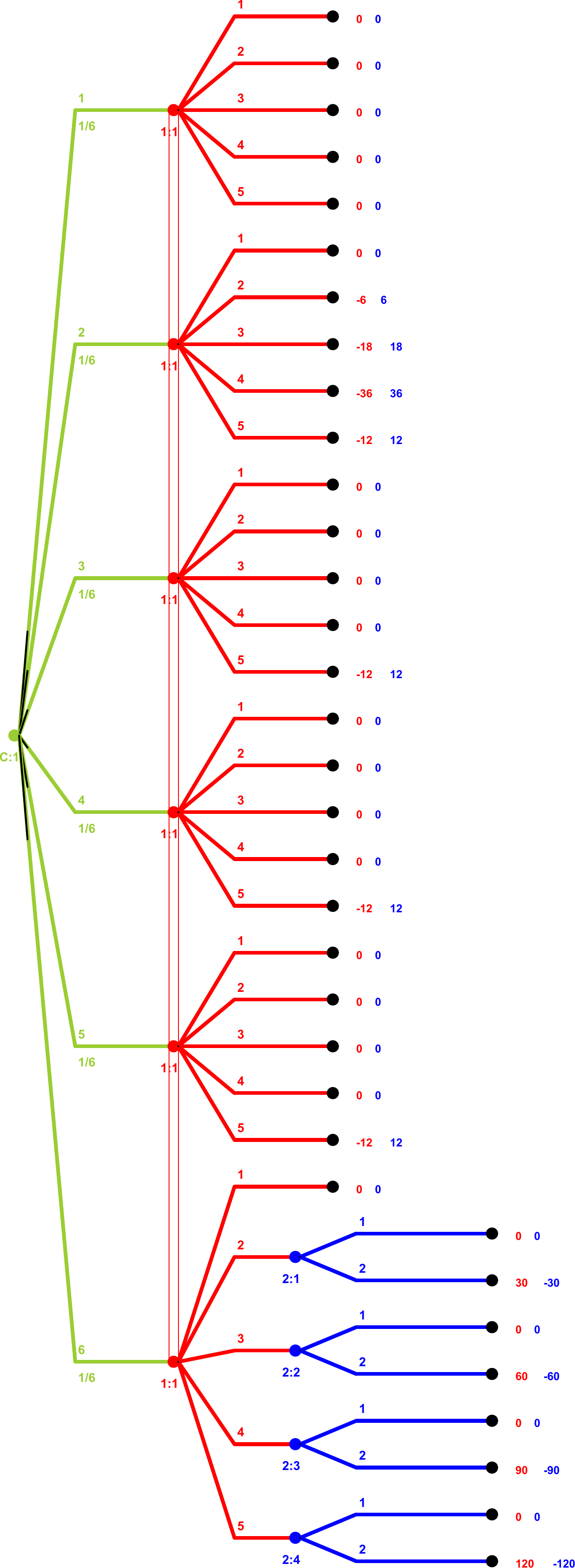}
    \hfill
    \includegraphics[width=0.43\textwidth]{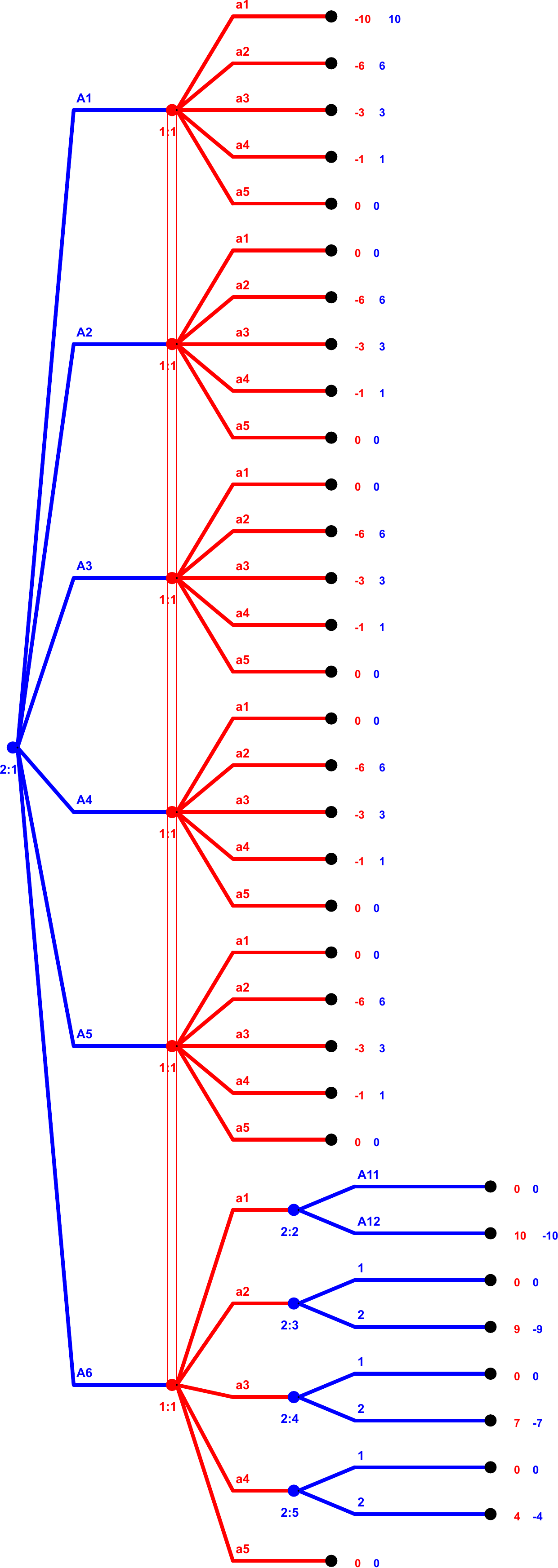}
    \caption{Example games to show the inability of common gadgets to reconstruct the whole Pareto set in the CDRNR setting. Left: Game 1 to break the max-margin gadget. Right: Game 2 to break the resolving gadget.}
    \label{fig:example_game}
\end{figure*}

\begin{figure}
    \centering
    \includegraphics[width=\linewidth]{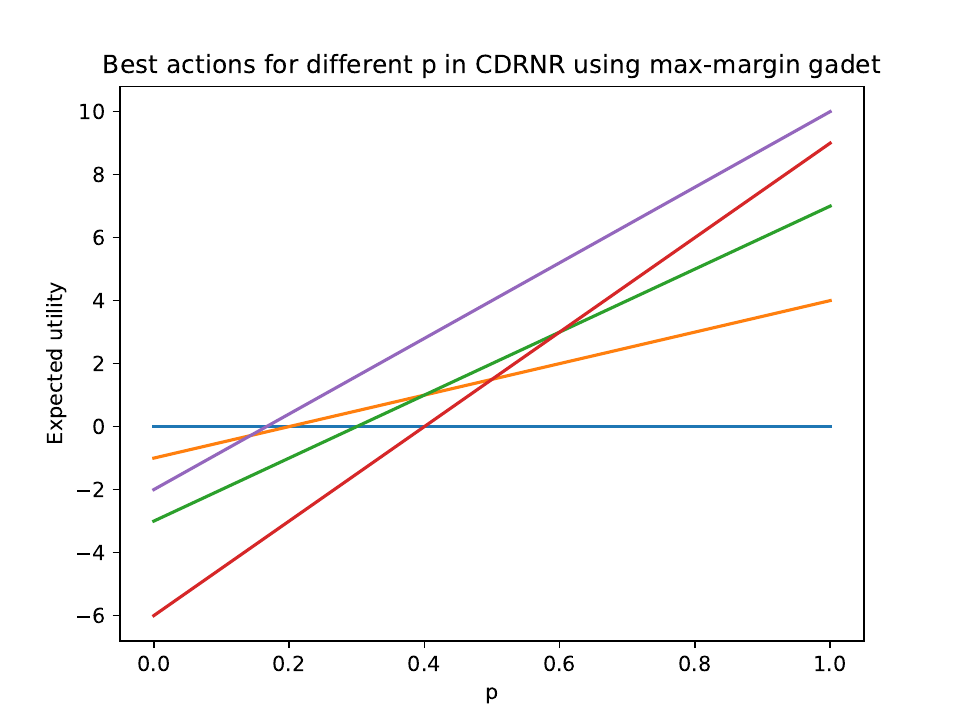}
    \includegraphics[width=\linewidth]{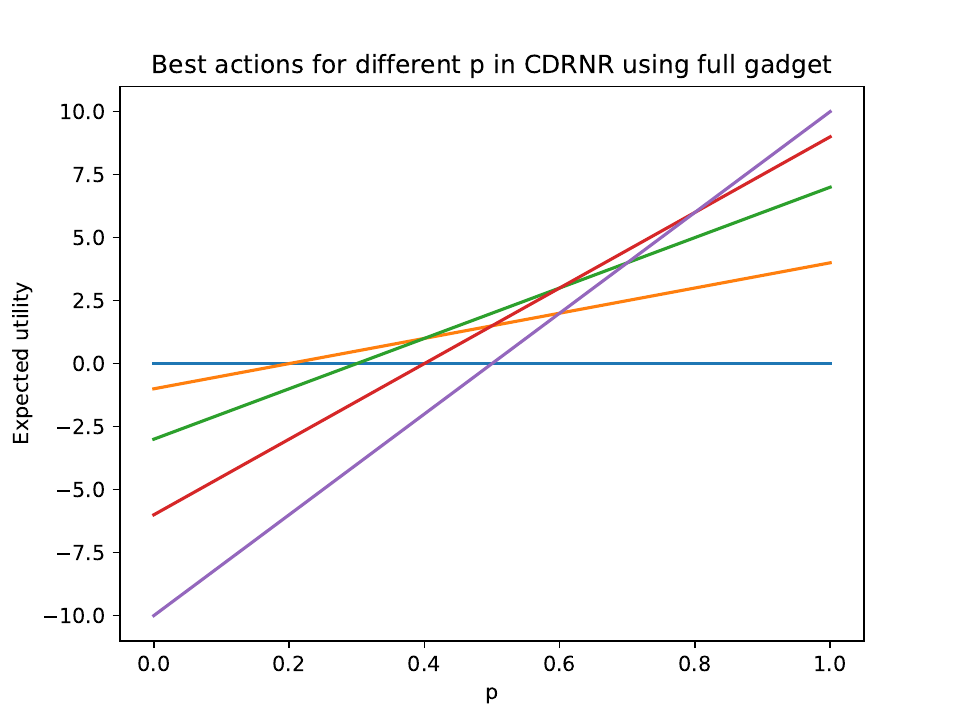}
    \includegraphics[width=\linewidth]{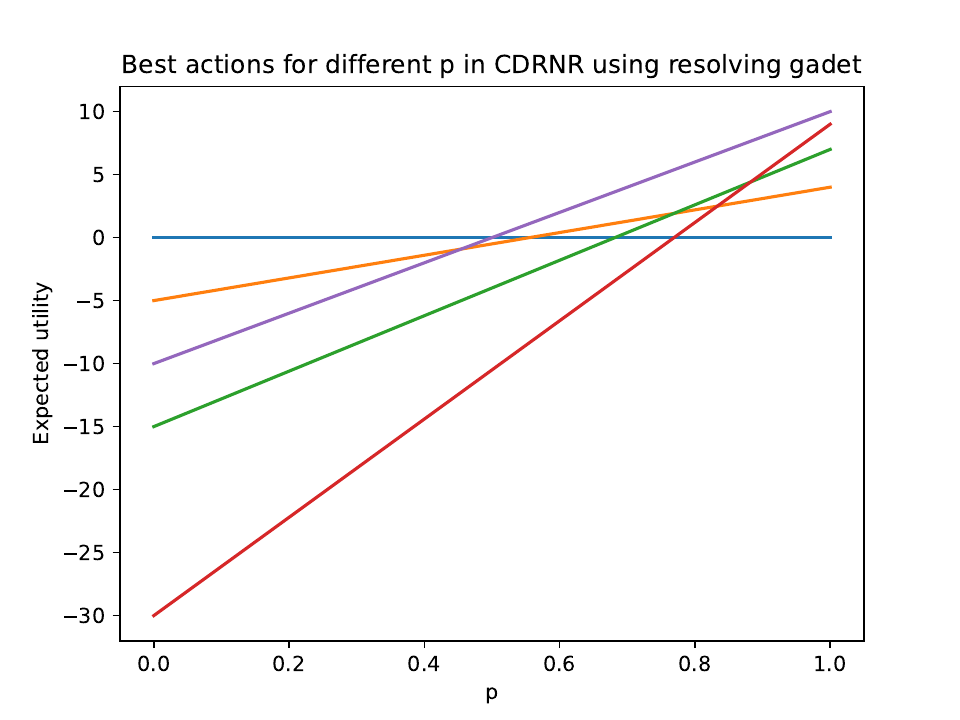}
    \caption{Value of each action based on $p$ in the games shown in Figure \ref{fig:example_game}, max-margin gadget in Game 1 (top), full gadget on both Game 1 and 2 (middle) and resolving gadget on Game 2 (bottom)}
    \label{fig:bestactions}
\end{figure}

\section{Experiment Details}
\subsection{Experimental Setup}
For all experiments, we use Python 3.8 and C++17. We solved linear programs using Gurobi 9.0.3, and experiments were done on an Intel i7 1.8GHz CPU with 8GB RAM. We used Leduc Hold'em, imperfect information Goofspiel 5, and Liar's dice for the smaller detailed experiments. We used Goofspiel 6 for the large experiment, and we only ran it against the strategy generated by the CFR with three iterations. We used the torch library for the neural network experiment. For most of the experiments, we wanted to solve the concepts perfectly with perfect value function, so we used LP and fixed the parts of the game that needed to be fixed. For the neural network experiment with CDBR, we used CFR+ to solve the subgame and the neural network as a value function. For the value function experiment in CDRNR, we used CFR+ with 1000 iterations to solve the game and CFR+ with 500 iterations as a value function in the subgames.

\subsection{Domain Definition}
\textbf{\textit{Leduc Hold'em}} is a medium-sized poker game. Both players give one chip at the beginning of the match and receive one card from a deck with six cards of 2 suits and three ranks. Then players switch and can call or bet. After a bet, the opponent can also fold, which finishes the game, and he forfeits all the staked money. After both players call or after at most two bets public card is revealed, and another betting round begins. In the first round, the bet size is two, and in the second, it is 4. If the game ends without anyone, folding cards are compared, and the player with pair always wins, and if there is no pair, the player with the higher card wins. If both have the same card, the money is split. \textbf{\textit{Goofspiel}} is a bidding card game where players are trying to obtain the most points. Cards are shuffled and set face-down. Both players have $K$ cards with values from 1 to $K$. These cards may be used as a bid. After bidding with that card, the player cannot play it again. Each turn, the top point card is revealed, and players simultaneously play a bid card; the point card is given to the highest bidder or discarded if the bids are equal. In this implementation, we use a fixed deck with K = 5 and K = 6. \textbf{\textit{Liar's Dice}} is a game where players have some number of dice and secretly roll. The first player bids rolled numbers, and the other player can either bid more or disbelieve the first player. When bidding ends with disbelief action, both players show dice. If the bid is right, the caller loses one die, and if the bidder is wrong, the bidder loses one die. Then the game continues, but for our computation, we use a version that ends with the loss of a die, and we use only a single die with four sides for each player.

\section{Proofs}
\begin{lemma}
Let $G$ be zero-sum imperfect-information extensive-form game with perfect recall. Let $\sigma_\ps^\text{F}$ be fixed opponent's strategy, let $T$ be some trunk of the game. If we perform CFR iterations in the trunk for player $\pr$ then for the best iterate $\hat{\sigma}_{\pr}$ $\max_{\sigma_{\pr}^* \in \Sigma_\pr} u_\pr(\sigma_{\pr}^*, \sigma_\ps^F)_V^T - u_\pr(\hat{\sigma}_{\pr}, \sigma_\ps^F)_V^T \leq \Delta\sqrt{\frac{A}{T}}|\mathcal{I}_{TR}| + N_S\epsilon_S$ where $\Delta$ is variance in leaf utility, $A$ is an upper bound on the number of actions, $|\mathcal{I}_{TR}|$ is number of information sets in the trunk, $N_S$ is the number of information sets at the root of any subgame and value function error is at most $\epsilon_S$.
\end{lemma}
\begin{proof}
Using Theorem~2 from \cite{burch2014solving} we know that regret for player $\pr$ is bounded $R_\pr^T = \frac{1}{T}\max_{\sigma_\pr^* \in \Sigma_\pr}\sum_{t=1}^T(u_\pr(\sigma_\pr^*, \sigma_\ps^F)_V^T - u_\pr(\sigma_\pr^t, \sigma_\ps^F))_V^T \leq \Delta\sqrt{AT}|\mathcal{I}_{TR}| + TN_S\epsilon_S$. Then we can directly map the regret to a different regret, that uses time-independent loss function $l(\sigma_\pr) = -u_\pr(\sigma_\pr, \sigma_\ps^F)_V^T$. We can then use Lemma~2 from \cite{lockhart2019computing} and we get $l(\hat{\sigma}_\pr) - \min_{\sigma_\pr^* \in \Sigma_\pr}l(\sigma_\pr^*) \leq \frac{R_\pr^T}{T}$. Substituting $l$ and $R_\pr^T$ back we get 
\[
    \max_{\sigma_{\pr}^* \in \Sigma_\pr} u_\pr(\sigma_{\pr}^*, \sigma_\ps^F)_V^T - u_\pr(\hat{\sigma}_{\pr}, \sigma_\ps^F)_V^T \leq \Delta\sqrt{\frac{A}{T}}|\mathcal{I}_{TR}| + N_S\epsilon_S
\]
\end{proof}

\begin{theorem}
Let $G$ be zero-sum extensive-form game with perfect recall. Let $\sigma_\ps^\text{F}$ be fixed opponent's strategy, let $\mathcal{P}$ be any subgame partitioning of the game $G$ and let $\sigma_\pr^{\sbr}$ be a CDBR given approximation $\bar{V}$ of the optimal value function $V$ with error at most $\epsilon_V$, partitioning $\mathcal{P}$ and opponent strategy $\sigma_\ps^\text{F}$ approximated in each step with regret at most $\epsilon_R$, formally $\sigma_\pr^{\sbr} = \sigma_\pr^\sbr(\sigma_\ps^F)_V^\mathcal{P}$. Let $\sigma^{NE}$ be any Nash equilibrium. Then $u_\pr(\sigma_\pr^{\sbr}, \sigma_\ps^F) + |\mathcal{S}|\epsilon_R + \sum_{S \in \mathcal{S'}}|I_{S^B}|\epsilon_V \geq u_\pr(\sigma^{NE})$.
\end{theorem}
\begin{proof}
Using subgame partitioning $P$, let $T_1$ be the trunk of the game. From the properties of a NE $u_\pr(\sigma^{NE}) \leq u_\pr(\sigma_\pr^{NE}, \sigma_\ps^\text{F}) \leq u_\pr(\sigma_\pr^{NE}, \sigma_\ps^\text{F})^{T_1}_V$. To compute CDBR we are maximizing in the trunk and using error in the value function and non-zero regret of the computed strategy $u_\pr(\sigma_\pr^{\sbr}, \sigma_\ps^\text{F})^{T_1}_V + |I_{T^B_1}|\epsilon_V + \epsilon_R \geq u_\pr(\sigma_\pr^{NE}, \sigma_\ps^\text{F})^{T_1}_V$. We continue using induction over steps with induction assumption that in step $i$, $u_\pr(\sigma_\pr^{\sbr}, \sigma_\ps^\text{F})^{T_i}_V + \epsilon_i \geq u_\pr(\sigma^{NE})$. We already know it holds for $T_1$. Now we assume we have trunk $T_{i-1}$ for which the induction step holds and trunk $T_i$ which is $T_{i-1}$ joined with new subgame $S_{i-1}$. Our algorithm recovers approximate equilibrium in $S_{i-1}$ using $\bar{V}$ at the boundary $S^B_{i-1}$, which means $u_\pr(\sigma^{\sbr}_\pr, \sigma_\ps^F)_V^{T_i} + |I_{S^B_{i-1}}|\epsilon_V + \epsilon_R \geq u_\pr(\sigma^{NE_{S_{i-1}}}_\pr \cup \sigma_\pr^{\sbr_{T_{i-1}}}, \sigma_\ps^F)_V^{T_i}$. If we use equilibrium for the opponent in the subgame $S_{i-1}$ we can replace equilibrium in the subgame by the value function $V$ and we have $u_\pr(\sigma^{NE_{S_{i-1}}}_\pr \cup \sigma_\pr^{\sbr_{T_{i-1}}}, \sigma_\ps^F)_V^{T_i} \geq u_\pr(\sigma_\pr^{\sbr}, \sigma_\ps^\text{F})^{T_{i-1}}_V$ and joining it all together we have $u_\pr(\sigma_\pr^{\sbr}, \sigma_\ps^\text{F})^{T_{i-1}}_V \leq \textbf{\textbf{}}u_\pr(\sigma_\pr^{\sbr}, \sigma_\ps^\text{F})^{T_i}_V + |I_{S^B_{i-1}}|\epsilon_V + \epsilon_R$ and $u_\pr(\sigma^{NE}) \leq \textbf{\textbf{}}u_\pr(\sigma_\pr^{\sbr}, \sigma_\ps^\text{F})^{T_i}_V + |I_{S^B_{i-1}}|\epsilon_V + \epsilon_R + \epsilon_{i-1}$. Accumulating the errors through the subgames will give the desired result $u_\pr(\sigma_\pr^{\sbr}, \sigma_\ps^F) + |\mathcal{S}|\epsilon_R + \sum_{S \in \mathcal{S'}}|I_{S^B}|\epsilon_V \geq u_\pr(\sigma^{NE})$ We omit last subgame from the accumulated value function error because the last step does not use value function.
\end{proof}

\begin{theorem}
Let $G$ be any zero-sum extensive-form game with perfect recall and let $\sigma_\ps^\text{F}$ be any fixed opponent's strategy in $G$. Then we set $G^M$ as restricted Nash response modification of $G$ using $\sigma_\ps^\text{F}$. Let $\mathcal{P}$ be any subgame partitioning of the game $G^M$ and using some $p \in \langle0,1\rangle$, let $\sigma_\pr^{\srnr}$ be a CDRNR given approximation $\bar{V}$ of the optimal value function $V$ with error at most $\epsilon_V$ and opponent strategy $\sigma_\ps^\text{F}$ approximated in each step with regret at most $\epsilon_R$, formally $\sigma_\pr^{\srnr} = \sigma_\pr^\srnr(\sigma_\ps^F, p)_V^\mathcal{P}$. Let $\sigma^{NE}$ be any Nash equilibrium in $G$. Then $u_\pr(\sigma_\pr^{\srnr}, \sigma_\ps^F) + \sum_{S \in \mathcal{S'}}|I_{S^O}|(1-p)\epsilon_V + |\mathcal{S}|\epsilon_R + \sum_{S \in \mathcal{S'}}|I_{S^B}|\epsilon_V \geq u_\pr(\sigma^{NE})$.
\end{theorem}
\begin{proof}
Let $T^M_1$ be a trunk of a modified game $G^M$ using partitioning $\partition$. We will use $u^G(\sigma)$ as utility in $G$. Utility of player $\pr$ for playing Nash equilibrium of the $G$ in trunk $T_1$ will be higher or the same as game value of $G$, formally $u^G_\pr(\sigma^{NE}) \leq u^{G^M}_\pr(\sigma^{NE})^{T^M_1}_V$. To compute CDRNR we use the approximate value function. In the fixed part of the game $G^F$ the value will be worse at most by sum of errors as in the CDBR case. However, in the $G'$ the situation is more complicated and we use Theorem~2 from \cite{burch2014solving} to bound the utility increase, resulting in $u^{G^M}_\pr(\sigma_\pr^{\srnr}, BR(\sigma_\pr^{\srnr}))^{T_i}_V + |I_{T^{M,B}_1}|\epsilon_V + |I_{T^M_1}|\epsilon_R \geq u^G_\pr(\sigma^{NE})$. We continue using induction over steps with induction assumption that in step $i$, $u^{G^M}_\pr(\sigma_\pr^{\srnr}, BR(\sigma_\pr^{\srnr}))^{T^M_i}_V + \epsilon_i \geq u^G_\pr(\sigma^{NE})$ and we already showed it holds for $T_1$. Now we assume we have trunk $T_{i-1}$ for which the induction step holds and trunk $T_i$ which is $T_{i-1}$ joined with new subgame $S_{i-1}$. Our algorithm recovers approximate equilibrium in $S^M_{i-1}$ and we want similar equation as for the CDBR. Part of the game tree $G^F$ has the errors bounded as in CDBR but because we use gadget in the $G'$ we need to also consider error in actions ending with value function player $\ps$ can play in the top with error bounded by $\epsilon_V$. We have $|I_{S^O}|$ of actions leading out of the tree so the error increase in the $G'$ going to the next subgame is at most $|I_{S^O}|\epsilon_V + |I_{S^{M,B}_{i-1}}|\epsilon_V + |I_{S^M_{i-1}}|\epsilon_R$ which together gives us $u^{G^M}_\pr(\sigma'_\pr, BR(\sigma'_\pr))_V^{T_i} + (1-p)|I_{S^O}|\epsilon_V + |I_{S^{M,B}_{i-1}}|\epsilon_V + |I_{S^M_{i-1}}|\epsilon_R \geq u^{G^M}_\pr(\sigma^\srnr_\pr, BR(\sigma^\srnr_\pr))_V^{T_{i-1}}$, where $\sigma_\pr'$ is a combination of the strategy we approximated in the subgame and the fixed strategy from previous step, formally $\sigma_\pr' = \sigma_\pr^{S_{i-1}} \cup \sigma_\pr^{\srnr, T_{i-1}}$. Joining it with the induction assumption we have $u^{G^M}_\pr(\sigma^\srnr_\pr, BR(\sigma^\srnr_\pr))_V^{T_i} + (1-p)|I_{S^O}|\epsilon_V + |I_{S^{M,B}_{i-1}}|\epsilon_V + |I_{S^M_{i-1}}|\epsilon_R + \epsilon_{i-1} \geq u^G_\pr(\sigma^{NE})$. Accumulating the errors in the last subgame we have $u^{G^M}_\pr(\sigma^\srnr_\pr, BR(\sigma^\srnr_\pr)) + \sum_{S \in \mathcal{S'}}|I_{S^O}|(1-p)\epsilon_V + |\mathcal{S}|\epsilon_R + \sum_{S \in \mathcal{S'}}|I_{S^B}|\epsilon_V$. However, we still need to show it works for $u^G_\pr(\sigma_\pr^{\srnr}, \sigma_\ps^F)$. We can do it by replacing strategy of player $\ps$ in the $G'$ by $\sigma_\ps^F$ which will effectively transform $G^M$ game back to $G$ with player $\ps$ playing $\sigma_\ps^F$. Since we did this transformation by changing the strategy that was a best response the utility can only increase and $u^G_\pr(\sigma_\pr^{\srnr}, \sigma_\ps^F) \geq u^{G^M}_\pr(\sigma^\srnr_\pr, BR(\sigma_\pr^\srnr))$ which concludes the proof.
\end{proof}

\begin{theorem}
Let $G$ be any zero-sum extensive-form game with perfect recall and let $\sigma_\ps^\text{F}$ be any fixed opponent's strategy in $G$. Then we set $G^M$ as restricted Nash response modification of $G$ using $\sigma_\ps^\text{F}$. Let $\mathcal{P}$ be any subgame partitioning of the game $G^M$ and using some $p \in \langle0,1\rangle$, let $\sigma_\pr^{\srnr}$ be a CDRNR given approximation $\bar{V}$ of the optimal value function $V$ with error at most $\epsilon_V$, partitioning $\mathcal{P}$ and opponent strategy $\sigma_\ps^\text{F}$, which is approximated in each step with regret at most $\epsilon_R$, formally $\sigma_\pr^{\srnr} = \sigma_\pr^\srnr(\sigma_\ps^F, p)_V^\mathcal{P}$. Then exploitability has a bound $\mathcal{E}(\sigma_\pr^\srnr) \leq \mathcal{G}(\sigma_\pr^\srnr, \sigma_\ps^F)\frac{p}{1-p} \sum_{S \in \mathcal{S'}}|I_{S^O}|(1-p)\epsilon_V + |\mathcal{S}|\epsilon_R + \sum_{S \in \mathcal{S'}}|I_{S^B}|\epsilon_V$, $\mathcal{E}$ and $\mathcal{G}$ are defined in Section~\ref{sec:background}.
\end{theorem}
\begin{proof}
We will examine the exploitability increase in each step. First, we define gain in a single step as $\mathcal{G}(\sigma_\pr, \sigma_\ps)_V^{T_i} = u_\pr(\sigma_\pr, \sigma_\ps)_V^{T_i} - u_\pr(\sigma_\pr, \sigma_\ps)_V^{T_{i-1}}$ for $i > 0$ and $\mathcal{G}(\sigma_\pr, \sigma_\ps)_V^{T_0} = u_\pr(\sigma_\pr, \sigma_\ps)_V^{T_0} - u_\pr(\sigma^{NE})$. This is consistent with full definition of gain because sum of gains over all steps will results in $u_\pr(\sigma_\pr, \sigma_\ps)^G - u_\pr(\sigma_\pr, \sigma_\ps)^{T_n}_V + u_\pr(\sigma_\pr, \sigma_\ps)^{T_n}_V - ... - u_\pr(\sigma_\pr, \sigma_\ps)^{T_0}_V + u_\pr(\sigma_\pr, \sigma_\ps)^{T_0}_V - u_\pr(\sigma^{NE}) = u_\pr(\sigma_\pr, \sigma_\ps)^G - u_\pr(\sigma^{NE}) = \mathcal{G}(\sigma_\pr, \sigma_\ps)$. We define exploitability in a single step similarly as $\mathcal{E}(\sigma_\pr)_V^{T_i} = u_\ps(\sigma_\pr, BR(\sigma_\pr))_V^{T_i} - u_\ps(\sigma_\pr, BR(\sigma_\pr))_V^{T_{i-1}}$ for $i > 0$ and $\mathcal{E}(\sigma_\pr)_V^{T_0} = u_\ps(\sigma_\pr, BR(\sigma_\pr))_V^{T_0} - u_\ps(\sigma^{NE})$ and it also sums to full exploitability. In each step we approximate the strategy in the modified game, having full utility in step written as $\mathcal{G}(\sigma_\pr^{\srnr}, \sigma^F_\ps)_V^{T^M_i}p - \mathcal{E}(\sigma_\pr^{\srnr})_V^{T^M_i}(1-p)$. If we had exact equilibrium in the subgame this would always be at least 0. However, we have $\bar{V}$ instead of $V$, values at the top of the gadget are not exact and the computed strategy has regret $\epsilon_R$. As in the previous proof the error is bounded by $|I_{S^O}|(1-p)\epsilon_V + |I_{S^{M,B}_{i-1}}|\epsilon_V + |I_{S^M_{i-1}}|\epsilon_R$ and we can write $\mathcal{G}(\sigma_\pr^{\srnr}, \sigma^F_\ps)_V^{T^M_i}p - \mathcal{E}(\sigma_\pr^{\srnr})_V^{T^M_i}(1-p) + |I_{S^O}|(1-p)\epsilon_V + |I_{S^{M,B}_{i-1}}|\epsilon_V + |I_{S^M_{i-1}}|\epsilon_R \geq 0$. We reorganize the equation to get $\mathcal{G}(\sigma_\pr^{\srnr}, \sigma^F_\ps)_V^{T^M_i}\frac{p}{1-p}  + |I_{S^O}|(1-p)\epsilon_V + |I_{S^{M,B}_{i-1}}|\epsilon_V + |I_{S^M_{i-1}}|\epsilon_R \geq \mathcal{E}(\sigma_\pr^{\srnr})_V^{T^M_i}$ and summing over all the steps gives us $\mathcal{G}(\sigma_\pr^{\srnr}, \sigma^F_\ps)\frac{p}{1-p} + \sum_{S \in \mathcal{S'}}|I_{S^O}|(1-p)\epsilon_V + |\mathcal{S}|\epsilon_R + \sum_{S \in \mathcal{S'}}|I_{S^B}|\epsilon_V \geq \mathcal{E}(\sigma_\pr^{\srnr})$
\end{proof}

\section{CDBR Against Nash Strategy}
\label{app:example}
\begin{figure}
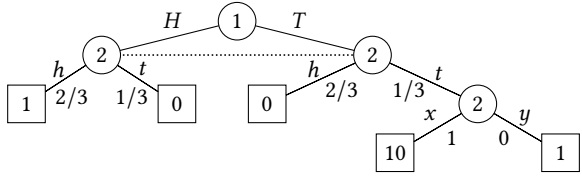

    \centering
    \begin{istgame}
    \setistOvalNodeStyle{.5cm}
    \setistRectangleNodeStyle{.5cm}
    \xtShowTerminalNodes[rectangle node]
    \istrooto(0){$\pr$}+{4.5mm}..{3.6cm}+
        \istb{H}[left, xshift=4mm, yshift=2.5mm]  \istb{T}[right, xshift=-4mm, yshift=2.5mm] \endist
    \istrooto(1)(0-1){$\ps$}+{6.5mm}..{2cm}+
        \istBt{h}[left, xshift=1mm, yshift=1.2mm]{2/3}[right, yshift=-2.5mm, xshift=-2.5mm]{1}[center]  \istBt{t}[right, xshift=-1mm, yshift=1.2mm]{1/3}[left, yshift=-2.5mm, xshift=2.5mm]{0}[center]  \endist
    \istrooto(2)(0-2){$\ps$}+{6.5mm}..{2.8cm}+
        \istBt{h}[left, xshift=1mm, yshift=1.2mm]{2/3}[right, yshift=-1.5mm, xshift=-1.5]{0}[center]  \istB{t}[right, xshift=-1mm, yshift=1.2mm]{1/3}[left, yshift=-0.8mm]  \endist
    \istrooto(6)(2-2){$\ps$}+{6.5mm}..{2.2cm}+
        \istBt{x}[left, yshift=1.5mm, xshift=1mm]{1}[right, yshift=-1.5mm]{10}[center]  \istBt{y}[right, yshift=1.5mm, xshift=-1mm]{0}[left, yshift=-1.5mm]{1}[center]  \endist
    \xtInfoset(0-1)(0-2)
    \end{istgame}
    \caption{Example of game where step best response is worse than NE against fixed strategy $\sigma(h)=\frac{2}{3}, \sigma(x)=1$.}
    \label{fig:sbr_nash_counterexample}
\end{figure}

\begin{observation}
An example in Figure~\ref{fig:sbr_nash_counterexample} shows that CDBR can perform worse than a Nash equilibrium against the fixed opponent because of the perfect opponent assumption after the depth-limit. An example is a game of matching pennies with a twist. Player $\ps$ can choose in the case of the tails whether he wants to give the opponent 10 instead of only 1. A rational player will never do it, and the equilibrium is a uniform strategy as in normal matching pennies.

Now we have an opponent model that plays $h$ with probability $\frac{2}{3}$ and always plays $x$. The best response to the model will always play $T$ and get payoff $\frac{10}{3}$. Nash equilibrium strategy will get payoff 2, and CDBR with depth-limit 2 will cut the game before the $x/y$ choice. Assuming the opponent plays perfectly after the depth-limit and chooses $y$, $\pr$ will always play $H$. Playing $H$ will result in receiving payoff $\frac{2}{3}$, which is higher than the value of the game ($\frac{1}{2}$) but lower than what Nash equilibrium can get against the model.
\label{obs:nash}
\end{observation}

\section{Additional Empirical Results}
\begin{figure}[h]
    \centering
    \includegraphics[width=\linewidth]{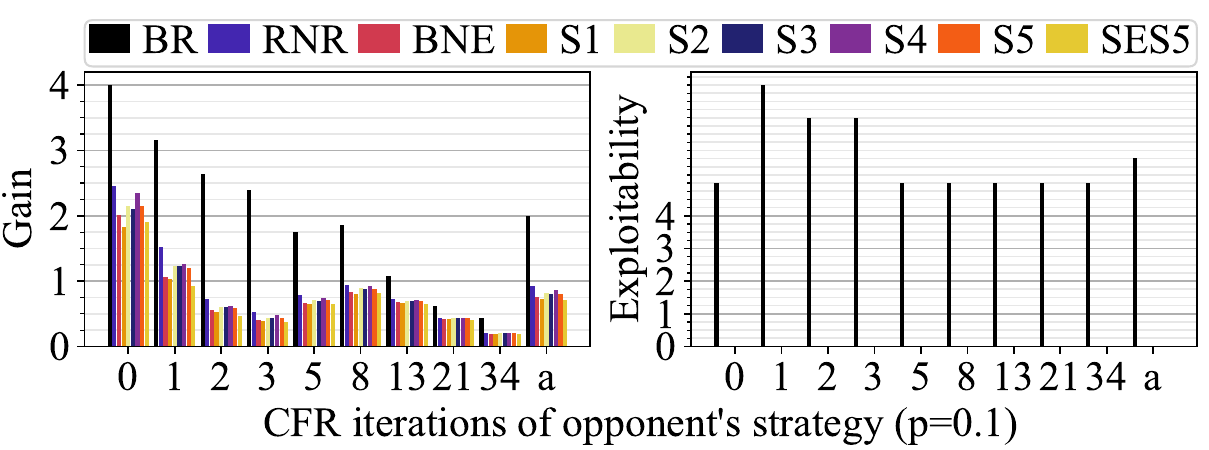}
    \includegraphics[width=\linewidth]{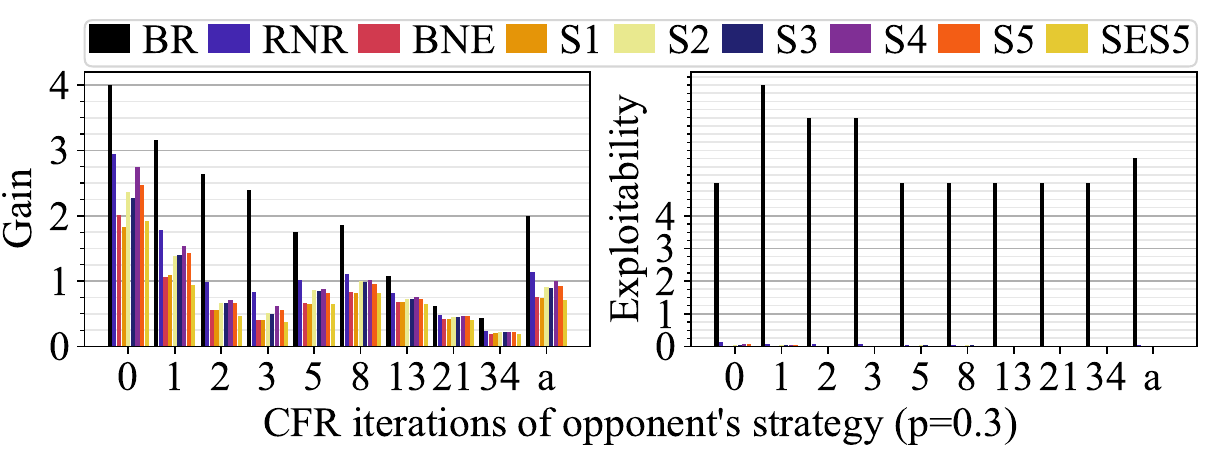}
    \includegraphics[width=\linewidth]{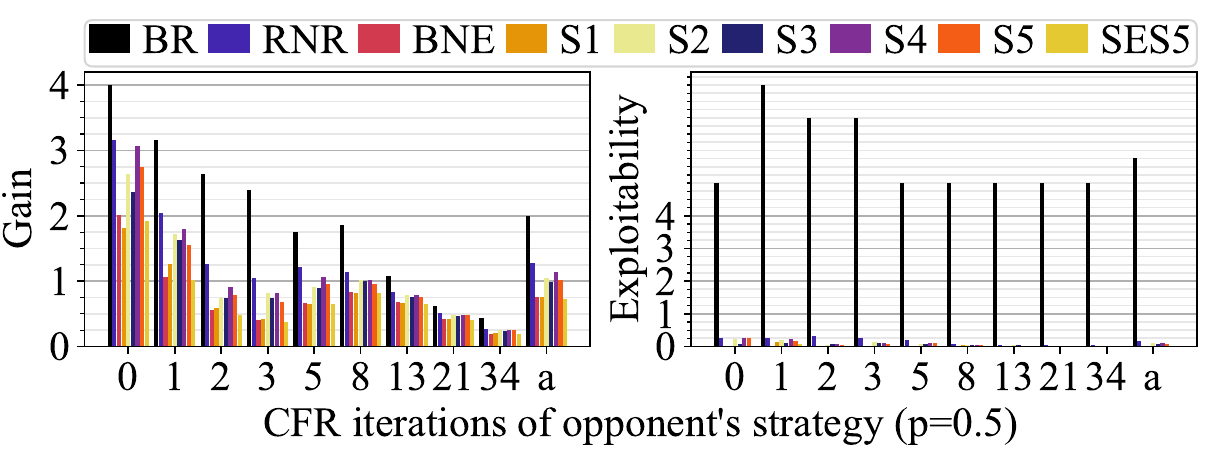}
    \includegraphics[width=\linewidth]{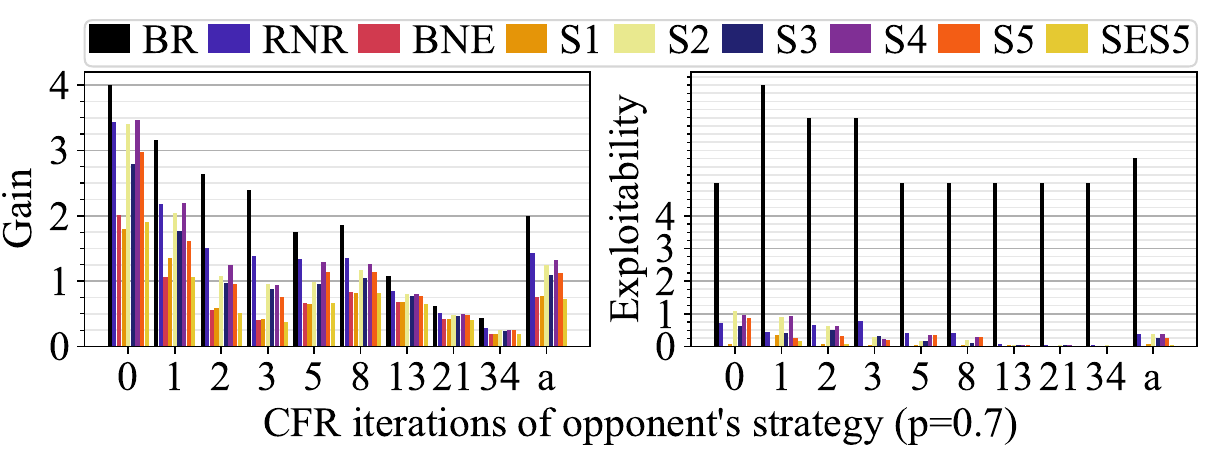}
    \includegraphics[width=\linewidth]{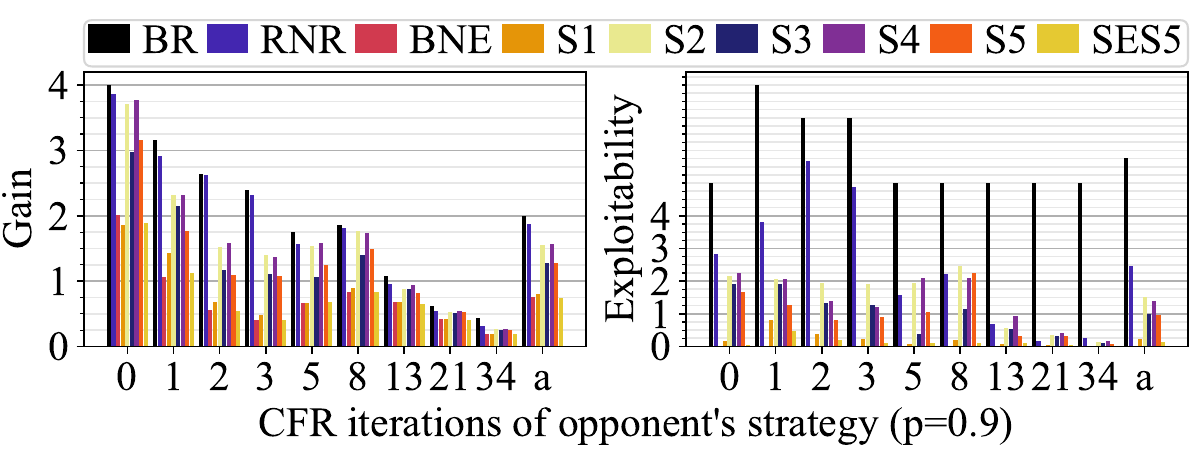}
    \caption{Additional results for CDRNR showing the performance of CDRNR with varying step-size. Generated on Goofspiel 5.}
\end{figure}

\paragraph{CDRNR}
We show more results for Goofspiel, Leduc Hold'em, and Liar's dice with different values of $p$. SX is CDRNR with step size denoted by X. We also evaluate SES and only use the highest step value of 5. Next, we show the same setup as in the main text with exactly the same partitioning as they used in SES, and we include more values of $p$.

\begin{figure}[h]
    \centering
    \includegraphics[width=\linewidth]{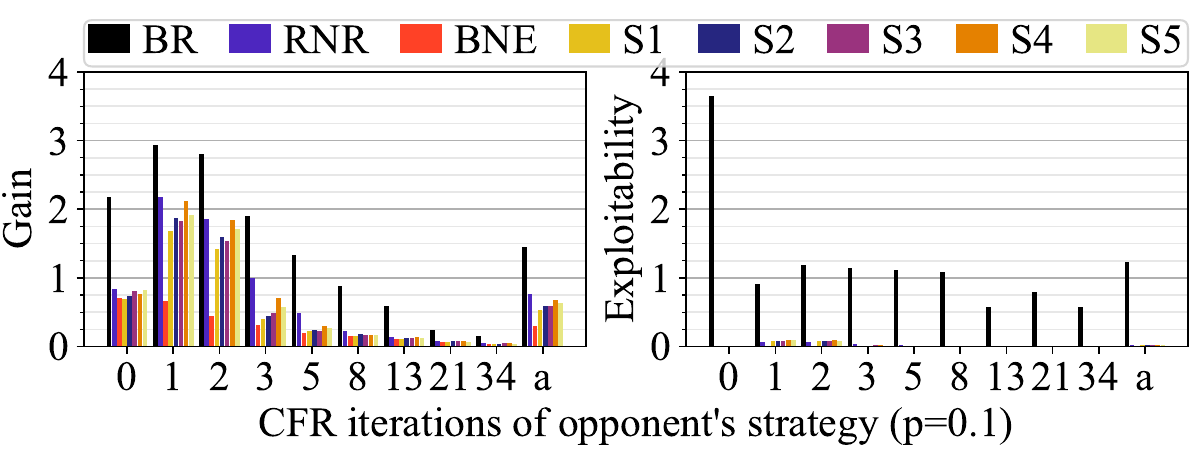}
    \includegraphics[width=\linewidth]{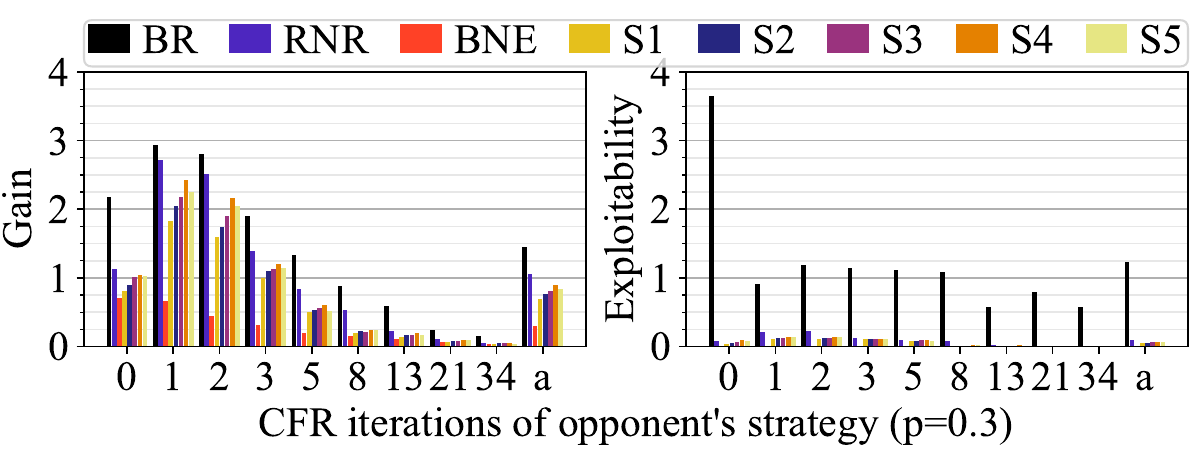}
    \includegraphics[width=\linewidth]{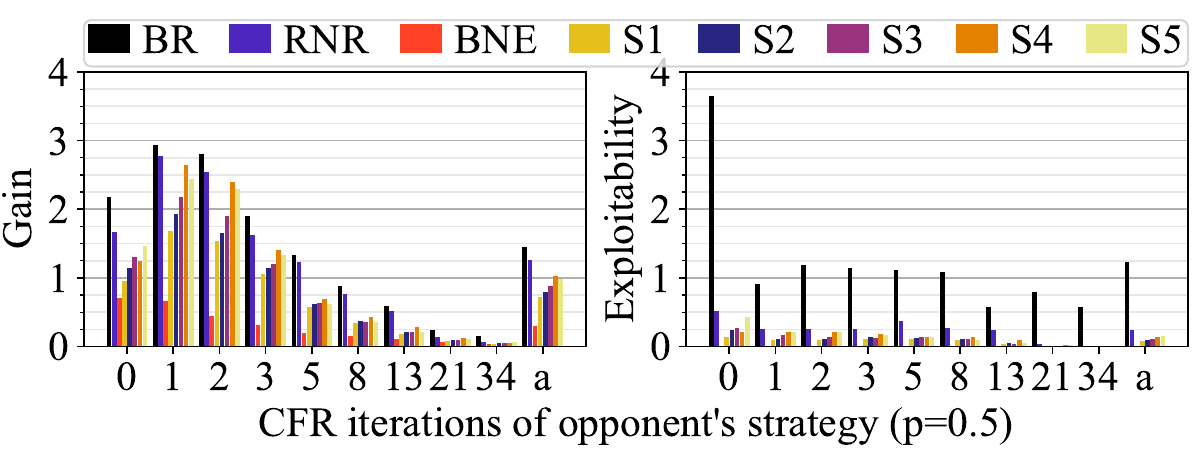}
    \includegraphics[width=\linewidth]{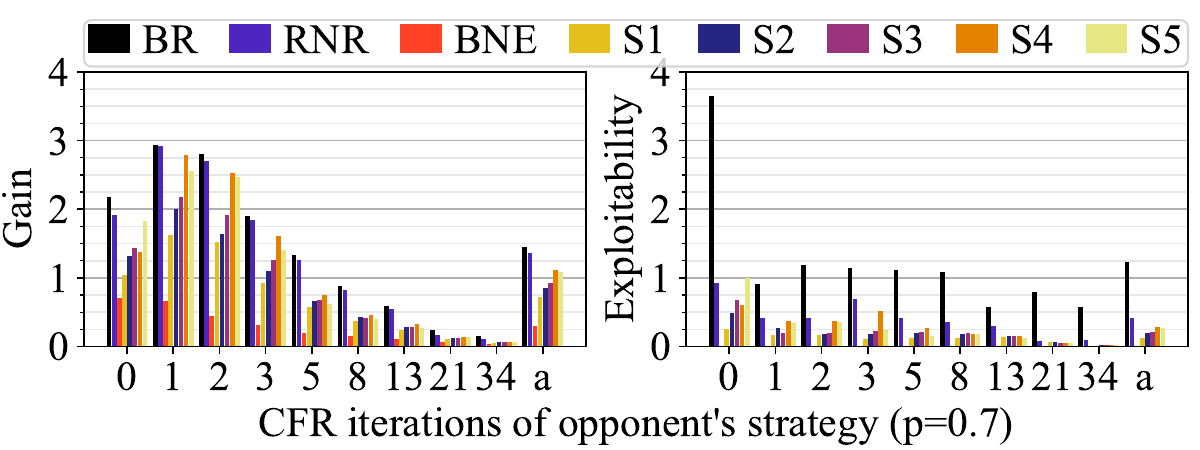}
    \includegraphics[width=\linewidth]{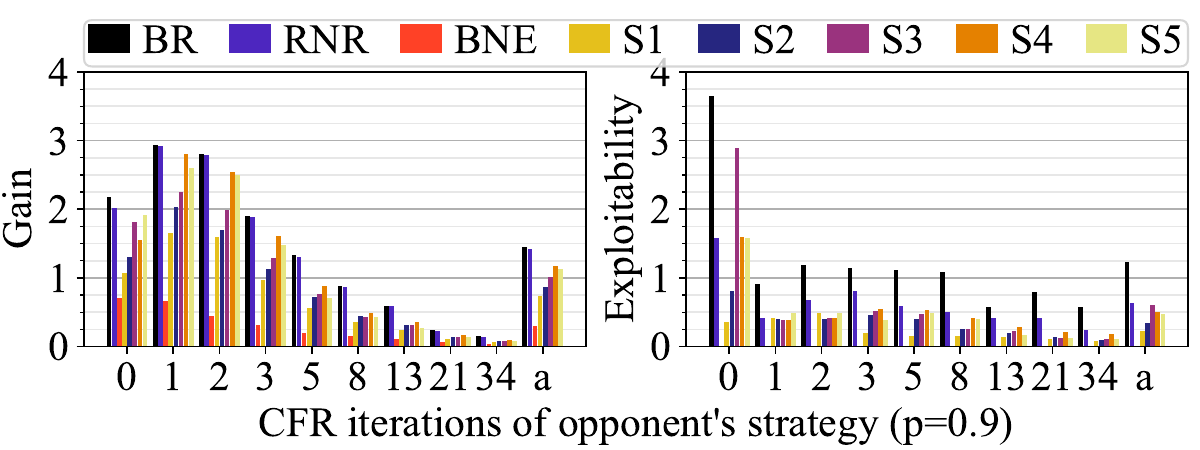}
    \caption{Additional results for CDRNR showing the performance of CDRNR with varying step-size. Generated on Leduc Hold'em.}
\end{figure}

\begin{figure}[h]
    \centering
    \includegraphics[width=\linewidth]{graphics/leduc_results_01.pdf}
    \includegraphics[width=\linewidth]{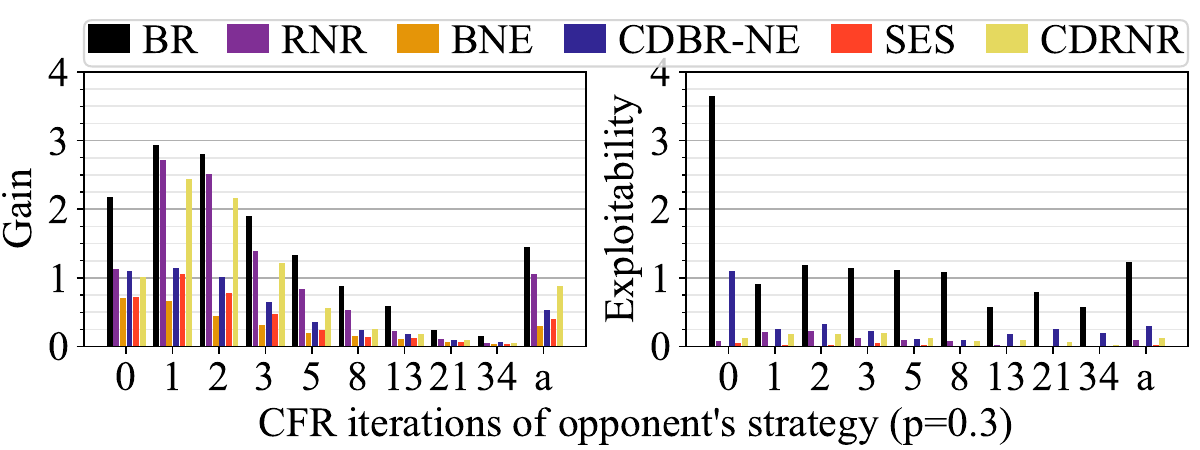}
    \includegraphics[width=\linewidth]{graphics/leduc_results_05.pdf}
    \includegraphics[width=\linewidth]{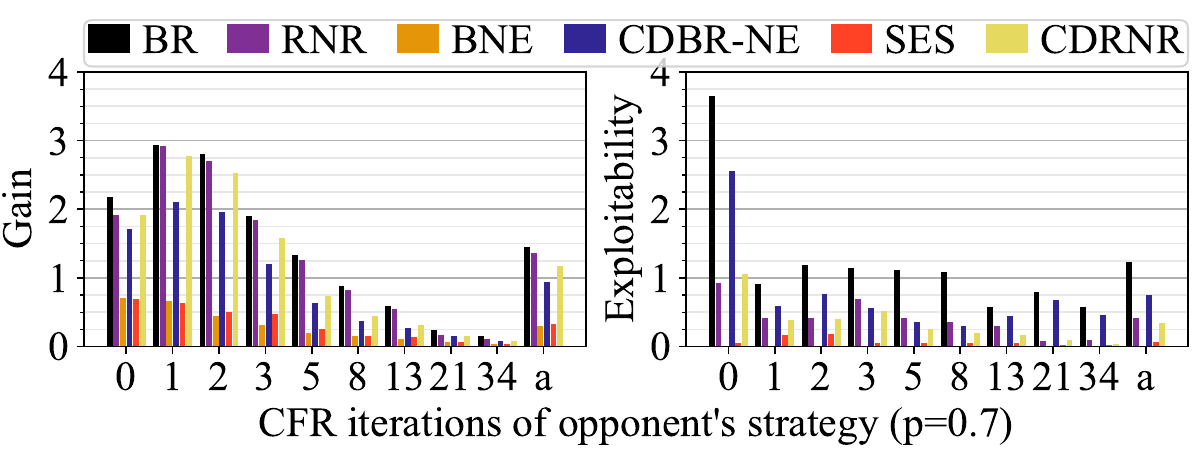}
    \includegraphics[width=\linewidth]{graphics/leduc_results_09.pdf}
    \caption{Additional results for CDRNR with different values of $p$. Generated on Leduc Hold'em split only by the round.}
\end{figure}

\begin{figure}
    \centering
    \includegraphics[width=\linewidth]{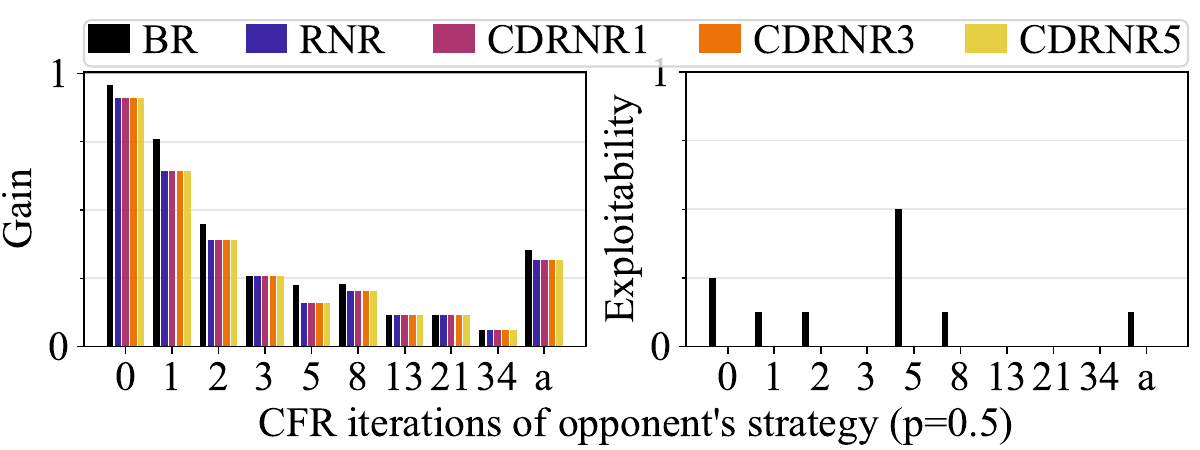}
    \caption{Additional result on Liar's dice. For every $p$ it exactly mimics the RNR so we only show one value.}
\end{figure}

\paragraph{Repeated RPS}
In Figure~\ref{fig:rps}, we show the strategy sets recovered for all possible $p$ against one strategy in two round biased RPS where after the round information is revealed. As we explained before, we can see that SES cannot gain anything in a game where only information imperfections are simultaneous moves. Exp-strat can exploit only the second round, and it gains half of the maximum, while the other algorithms can gain the maximum and are more or less successful in achieving the best trade-off. The full gadget is the best, followed by the other gadgets without theoretical guarantees, and then by a combination of Nash and CDBR.
\begin{figure}
    \centering
    \includegraphics[width=\linewidth]{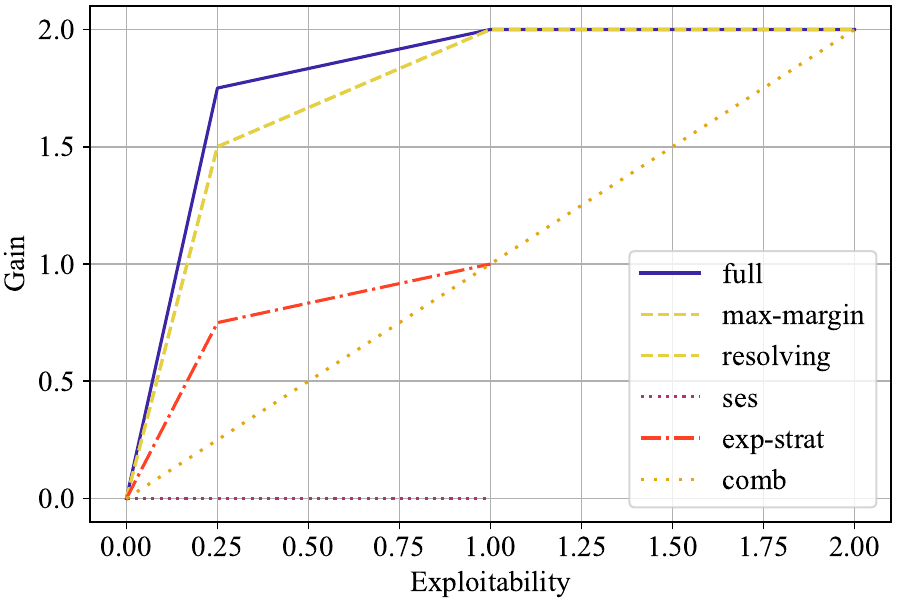}
    \caption{Results showing gain and exploitability trade-off in two round biased RPS. Max-margin and resolving gadget overlaps}
    \label{fig:rps}
\end{figure}

\section{SES Bound}
The bound in SES \cite{liu2022safe} is

\begin{theorem}
    Let $\mathbb{S}$ be a disjoint set of subgames S. Let $\sigma^* = \langle \sigma^*_1, \sigma^*_2\rangle$  be the NE where player \pr's strategy is constrained to be the same with $\sigma_\pr$ outside $\mathbb{S}$. Define $\Delta = \max_{S \in \mathbb{S}, I_\ps^i \in S_{top}}|CBV_\ps^{\sigma^*_\pr}(I_\ps^i) - v_\ps^\sigma(I_\ps^i)|$. Let $\tilde{p}(I_\ps^i)$ be the reach probability given by $\sigma_\ps^*$. Let $\hat{p}(I_\ps^i)$ be the estimation of reach probability $p(I_\ps^i)$ given by the real opponent strategy. Define $\tau = \max_{S \in \mathbb{S}, I_\ps^i \in S_{top}}|\frac{\hat{p}(I_\ps^i) - \tilde{p}(I_\ps^i)}{\tilde{p}(I_\ps^i)}|$. Whenever $1 - (2\tau + 1)\alpha > 0$, the exploitability bound is given by:
\end{theorem}

\[
    \mathcal{E}(\sigma'_\pr) \leq \mathcal{E}(\sigma^*_\pr) + \frac{2}{1 - (2\tau + 1)\alpha}\Delta
\]

We switched the players since authors in the previous work use player $\ps$ as the rational player.
 
We can see that the bound relies on the estimation being close to an equilibrium strategy defined by authors  as $\tau$. However, it does max over all the differences in reaches to the subgame, and in practice, some of the reaches will be very different, resulting in a large value of $\tau$. To demonstrate the difference, we assume the opponent model plays such that some action difference from equilibrium is 1, which is the highest it can be, and hence $\tau = 1$. Parameter $\alpha$ in SES directly matches $p$. For $\alpha = 0$, the bound is the same as in the max-margin gadget, and $\tau$ is disregarded. However, as $\alpha$ increases, the bound steeply rises, and as $\alpha$ goes in the limit to $\frac{1}{3}$, the bound goes to infinity, and for any larger $\alpha$, the bound says nothing. In comparison, our bound does not have this problem, and in the same setup, with $p = 0.5$, our bound still limits the exploitability by exactly the gain achieved. Note that since in SES, they do not account for errors in value function and errors in resolving, for this comparison only, we also omitted error terms caused by those errors.
\fi

\end{document}